\documentclass{llncs}
\usepackage{listings}
\usepackage{amssymb,amsfonts,amsmath}    % needed for including graphics e.g. EPS, PS
\usepackage{graphicx}    % needed for including graphics e.g. EPS, PS
\usepackage{epsfig}
\usepackage{lhelp}
\usepackage[show]{ed}
\usepackage{multirow}
\usepackage{verbatim}
\usepackage{latexsym}  
\usepackage{url}      
\usepackage{mathrsfs}
\usepackage{enumerate}
\usepackage[ruled,vlined,linesnumbered]{algorithm2e}
\usepackage{pbox}
\def\Q {\ensuremath{\mathbb{Q}}}

\newtheorem{Notation}{Notation}
\newtheorem{Theorem}{Theorem}
\newtheorem{Proposition}{Proposition}
\newtheorem{Definition}{Definition}
\newtheorem{Lemma}{Lemma}
\newtheorem{Corollary}{Corollary}

\newtheorem{Remark}{Remark}
\newtheorem{Example}{Example}

\newcommand{\MI}[1]{\mbox{{\rm MRI}$(#1)$}}
\newcommand{\pii}{{{\tt PlainInvInterp}}}
\newcommand{\mii}{{{\tt ModpInvInterp}}}

\newif\ifcomment\commentfalse

\newif\ifdraft
\draftfalse

\newif\ifbibtex
\bibtextrue

\newif\ifheight
\heightfalse

\newif\ifCoRR
\CoRRtrue

\newcommand{\wm}{{weakly multiplicatively}}

\renewenvironment{proof}
{{\bf {\sc Proof}:} }
   {$\triangleup$\\}
\renewenvironment{proof}
{{\bf {\sc Proof}:} }
  {$\square$\\}

\begin{document}

\ifCoRR

\title{Generating Program Invariants via Interpolation}

\else

\title{Generating Loop Invariants via Polynomial Interpolation}

\fi
% Polynomial Interpolation?

\author{Marc Moreno Maza\inst{1} \and Rong Xiao\inst{2}} % \inst{2}
\institute{
University of Western Ontario, Canada  ~\email{moreno@csd.uwo.ca}
\and 
University of Western Ontario, Canada ~\email{rong@csd.uwo.ca}}
% \email{moreno@csd.uwo.ca,rong@csd.uwo.ca}
% \thanks{Maplesoft and  MITACS-NCE, Canada}

% \author[Moreno Maza]{Marc~Moreno~Maza}
% \address{University of Western Ontario, Canada}
% \email{moreno@csd.uwo.ca}

% \author[Xiao]{Rong~Xiao}
% \address{University of Western Ontario, Canada}
% \email{rong@csd.uwo.ca}

%%% ----------------------------------------------------------------------
\maketitle
%%% ----------------------------------------------------------------------
%\tableofcontents

% \input{abstract}

\begin{abstract}
This article focuses on automatically generating
polynomial equations that are 
inductive loop invariants of computer programs.
We propose a new algorithm for this task, 
which is based on polynomial interpolation.
Though the proposed algorithm is not complete, 
it is efficient and can be applied to a broader range 
of problems compared to existing methods targeting similar problems.
% The proposed algorithm applies to two models 
% of computation, one deterministic and one non-deterministic.
The efficiency of our approach
is testified by experiments on a large
collection of programs.
The current implementation of our method is based on dense interpolation,
for which a total degree bound is needed. On the theoretical front, 
we study the degree and dimension of the invariant ideal
of loops which have no branches and where   
the assignments define a $P$-solvable recurrence.
% In the case of loop bodies involving  
% assignment operations only (thus no branches) such that  
% the assignments define a $P$-solvable recurrence,
% we study the degree and dimension of the loop invariant ideal.
% We provide degree bound for the polynomial invariants.
% We also provide dimension analysis for the invariant ideal, 
In addition, we obtain sufficient conditions 
for non-trivial polynomial equation invariants to exist (resp. not to exist).
\end{abstract}

\section{Introduction}

Many researchers have been using computer algebra 
to compute polynomial loop invariants, 
see for instance~\cite{SSM04,TRSS01,RK04a,RK04b,MS04a,MS04b,Kap05,CXYZ07,CJ07,RK07a,RK07b,Kov08}.
In this article, we propose an alternative method, based on 
interpolating polynomials at finitely many points on 
the reachable set of the loop under study.
This interpolation process\footnote{
Note that polynomial interpolation is different from
the interpolation in~\cite{KV09b}, which is called
\emph{Craig interpolation} in first order logic.}
 yields ``candidate loop invariants'' 
which are checked  by a new criterion based on 
polynomial ideal membership testing. 
% where
% the polynomial ideal membership  
% polynomial
% system solving.
% order to generate polynomial equations 
% that are inductive invariants of the loop. 

% OUR COMMENTS:
% After checking the above papers by M. Mueller-Olm and H. Seidl,
% M. Kauers and B. Zimmermann, L. Kovács (of which we were NOT aware
% at the time of writing our ISSAC submission) we believe that
Our paper proposes the following original results.
On the theoretical front, 
for  $P$-solvable loops with no branches, we supply a sharp degree 
bound (Theorem~\ref{thm: shape-degree})
for the invariant ideal, as well as dimension analysis (Theorem~\ref{thm: dimension}), 
of the invariant ideal.
We establish a new criterion (Corollary~\ref{coro:criterion})
based on polynomial system
solving for checking whether or not a given conjunction of polynomial equations
is indeed a loop invariant. 
Meanwhile, Corollary~\ref{cor: not-exist} states a 
sufficient condition for the invariant ideal of a loop
to be trivial.

On the algorithmic front, we propose a modular method 
(Algorithm~\ref{alg: mii})
for generating polynomial
       loop invariants. Thanks to polynomial interpolation, most of our
       calculations reduce to linear algebra.
As a consequence, the proposed method works in time $n^{d^{\mathcal{O}(1)}}$,
       where $n$ is the number of loop variables and $d$ is the total degree of 
       polynomials to interpolate.
       % while the other popular methods. All other methods
       % (listed by the referee) rely on Gr\{"}obner bases or 
       % quantifier elimination (via cylindrical algebraic decomposition).
       % For checking that a polynomial is as a loop invariant, we use 
       % triangular decomposition techniques, which are practically 
       % more efficient in this context, where ideals are of positive dimension.

 Our method is probabilistic and may not compute the whole 
       invariant ideal. However, the implementation (in Maple) 
of our method   computes all the invariants given 
         in each example proposed by Enric Rodr\'iguez Carbonell 
         on his page\footnote{\url{ http://www.lsi.upc.edu/~erodri/webpage/polynomial_invariants/list.html}}.
  Moreover, the degree and dimension estimates can help certifying that 
         whole invariant ideal has been obtained. For instance 
         in co-dimension one, the invariant ideal is necessarily principal.
        %% The main goal of this submission is to provide a practical
        %% alternative method to loop invariant automatic generation. 

% Another feature targeting efficiency: 
Our method needs not to 
        solve the recurrence relations associated with the loops
        and thus does not need to manipulate the algebraic numbers
        arisen as eigenvalues of these recurrence relations.
        Therefore, all polynomials and matrices involved
        in our method have their coefficients in the base field.
Our method applies to all loops which can be modeled as
algebraic transition systems~\cite{SSM04} and can be
generalized to handle loops which can be modeled as
       semi-algebraic transition systems~\cite{CXYZ07}. 
It can be applied to compute all 
       kinds of invariants 
\ifCoRR
(see the notions of different loop invariants presented 
in Section~\ref{sec:preliminaries}), 
\else
(see the notions of different loop invariants presented 
in Section~\ref{sec:preliminaries} of 
our technical report~\cite{DBLP:journals/corr/abs-1201-5086}),
\fi
which is not the case for the methods
       based on "recurrence solving" ~\cite{RK07a,Kov08}. % (e.g. L. Kovocs TACAS2008). 
In particular,  the methods in~\cite{RK07a} and~\cite{Kov08} 
       % [E. Rodriguez-Carbonell and D. Kapur ISSAC04], 
       % [L. Kovocs TACAS08] 
       apply only to compute absolute inductive invariants,  
       that is, the loop guard and branch conditions are
ignored (thus the loop goes to the branches randomly).
\ifCoRR
       This means that those methods can not find loop 
       invariants which are not absolute inductive invariants.
\fi
       % y1* x2 + y2 + y3 = x1 for the following loop:
%        y1 := 0;
%        y2 := 0;
%        y3 := x1;
%        # TARGET LOOP
%        while (y3 <> 0) do
%          if (y2 + 1 = x2)
%          then
%            y1 := y1 + 1;
%            y2 := 0;
%            y3 := y3 - 1;
%          else
%            y2 := y2 + 1;
%            y3 := y3 - 1;
%          end if;
%        end do;

        % As mentioned, our method succeeds for the examples posted 
        % by Enric Rodr\'iguez Carbonell on his webpage.
Our implementation is tested against 
the implementation of $3$ 
other methods~\cite{RK07a,RK07b,Kov08}
which were kindly made available to us by their authors.
The experimental results are shown in Section~\ref{sec:construction}.
While the performance of our method is comparable to that of \cite{RK07b}
on the tested (somehow simple) examples,
our method has less restrictive specifications than
the methods of~\cite{RK07a,Kov08} which target
only on absolute polynomial loop invariants for $P$-solvable loops/solvable mappings.
In addition, the method in~\cite{RK07a} applies only
when all assignments are invertible and the
eigenvalues of coefficient matrix for the 
linear part are positive 
\ifCoRR
rationals, and it is
claimed to be complete, that is, to compute all absolute polynomial loop invariants,
 when it is applicable. 
\else
rationals.
\fi
The method in \cite{Kov08}  can be applied all $P$-solvable 
loops in theory and is complete for loops without branches.
However, in the implementation,  the assignment
in the loop are required to be not ``coupled'' together.

Let us conclude this introduction
with a brief review of other works on loop invariant computation.
In~\cite{Karr76}, linear equations as invariants of a linear program at 
each location is considered, by tracking the reachable states
with a  method based on linear algebra. 
In \cite{MS04b}, the method in~\cite{Karr76}
is improved and generalized to generate polynomial 
equations as invariants; it is also
shown there that checking whether or not a linear equation is invariant 
is undecidable in general.
 In \cite{MS04a}, for polynomial programs, 
the Authors discuss 
methods  based on abstract 
interpretation, on checking 
whether or not a given polynomial equation is invariant,
as well as generating all polynomial invariants to a given total degree.
In \cite{RK04b,RK07b}, a different abstract interpretation
 technique is developed, which uses polynomial ideal operations 
(e.g. intersection, quotient)
as widen operators. In \cite{Kap05} and \cite{CXYZ07}, 
quantifier elimination techniques are used to infer invariants 
from a given template;
these methods requires expertise on supplying meaningful templates, while 
the complexity of quantifier elimination also restrict their practical efficiency.
In~\cite{SSM04}, Gr\"obner basis together with linear constraint solving is used to 
infer polynomial equations as invariants.

\section{Preliminaries}
\label{sec:preliminaries}
Let $\mathbb{Q}$ denote the rational numbers and
$\overline{\mathbb{Q}}$ the algebraic closure of $\mathbb{Q}$.
Let $\mathbb{Q}^*$ (resp. $\overline{\mathbb{Q}}^*$) denote the non zero elements in $\mathbb{Q}$ (resp. $\overline{\mathbb{Q}}$) .

\ifCoRR
\subsection{Notions on loop and loop invariants}

We will use the following simple loop (in {\sc Maple}-like syntax) to
introduce some notions related to loop and loop invariant
that we are going to use.

\begin{center}
    \fbox{
     \parbox[b]{0.2\textwidth} {
    {\bf 
\begin{tabbing}
% $X \in \mathfrak{T}_0$; \\
$x := a$;\\
$y := b$;\\
whil\=e $ x<10$ do \\
 \> $x := x+y^5$;\\
 \> $y := y+1$;\\
end do;
\end{tabbing}
}}
}
\end{center}

A \emph{loop variable} of a loop is a 
variable that is either
updated in the loop; or used to initialize/update the values of  other loop variables, e.g. $x,y,a,b$ are loop variables. Without loss of generality, 
we assume that all variables take only rational number values, i.e. from $\mathbb{Q}$.
By \emph{initial values} of a loop, we mean all possible tuples of 
the loop variables before executing the loop; the set of 
the initial values of the above loop is 
$$\{(x,y,a,b)\mid x=a, y=b, (a, b) \in \mathbb{Q}^2\}.$$ 
Given an initial value $\vec{v}$, 
the \emph{trajectory} of the loop starting at   $\vec{v}$, 
is the sequence of all tuples of 
of loop variable values %reachable 
\emph{at each entry} of the loop during the execution, with 
the loop variable being initialized by  $\vec{v}$; 
the trajectory of the above loop starting at $(x,y,a,b)=(1,0,1,0)$ is 
$$(1,0,1,0), (1,1,1,0),(2,2,1,0),(34,3,1,0).$$
 The collection of value tuples of all trajectories 
is called \emph{reachable set} of the loop. 
Note that, in general, it is  hard to describe 
a reachable set of a loop precisely. 
A \emph{loop invariant} (or \emph{plain loop invariant})
of a loop is a condition on the loop
variables satisfied by all the values  in the reachable set of the loop.

By  \emph{inductive reachable set} of a loop, we mean the 
reachable set of the loop while ignoring the guard condition, 
while by \emph{absolute reachable set} of a loop, we mean the 
reachable set of the loop while ignoring the guard conditions, the branch
conditions and viewing branches to be selected randomly. Then, by 
an \emph{inductive (loop) invariant} (resp. \emph{absolute (loop) invariant}) 
of a loop is a condition on the loop
variables satisfied by all the tuple values  
in the inductive (resp. absolute) reachable set of the loop.

It is easy to deduce that % for the loops in our consideration, 
an absolute  invariant is always an inductive invariant, and 
an inductive invariant is always a loop invariant. 
In principle, absolute inductive invariants are easier 
to study and compute than the inductive invariants and plain invariants.
However, the absolute invariants 
can be trivial, which is not of practical interest to 
program analysis. See the following example~\cite{SSM04} for 
the case of a trivial absolute invariant while inductive loop invariants
in not trivial.
\begin{center}
    \fbox{
     \parbox[b]{0.2\textwidth} {
    {\bf 
\begin{tabbing}
% $X \in \mathfrak{T}_0$; \\
$y_1 := 0$; \\
$y_2 := 0$; \\
$y_3 := x_1$; \\
whil\=e $y_2 \neq 0$ do \\
    \> if $y_2+1 = x_2$ \\
     \> then\= \\
    	 	\> \>   $y_1 := y_1 + 1$; \\
		\> \>   $y_2 := 0$; \\
		\> \>   $y_3 := y_3 - 1$; \\
       \> else \\ 
		\> \>   $y_2 := y_2 + 1$; \\
		\> \>   $y_3 := y_3 - 1$; \\
       \> end if\\
end do
\end{tabbing}
}}
}
\end{center}
 Indeed, the condition $ y_1 x_2 + y_2 + y_3 = x_1 $ is an inductive invariant of the above loop.
Note there are also loop invariants which are not inductive invariants, e.g. $x-1 =0$
is an invariant but not an inductive of the following loop.
\begin{center}
    \fbox{
     \parbox[b]{0.2\textwidth} {
    {\bf 
\begin{tabbing}
% $X \in \mathfrak{T}_0$; \\
  $x := 1$; \\
whil\=e $x \neq 1$ do \\
   \> $ x  := x + 1$; \\
end do
\end{tabbing}
}}
}
\end{center}

On the other hand, the inductive invariants are less likely to be trivial 
and easier to handle than the loop invariants.
\fi
In this article, we are interested in the inductive invariants
% and absolute inductive invariants 
that are given by polynomial equations
and that we call \emph{polynomial equation invariants}, 
or simply polynomial invariants when there is no possible confusion.
It is not hard to deduce that all polynomials that 
are inductive invariants (or loop invariants, or absolute invariants) 
of a loop form an ideal (which is indeed the ideal of the points in
the inductive reachable set), one can also refer~\cite{RK04}
for an alternative proof. We call the ideal of  polynomials 
which are inductive invariants of a given loop
the \emph{invariant ideal} of the loop.

In this paper, 
we consider loops % (or recurrences which can be simulated by loops) 
of the following shape.

\begin{center}
    \fbox{
     \parbox[b]{0.2\textwidth} {
    {\bf 
\begin{tabbing}
% $X \in \mathfrak{T}_0$; \\
whil\=e $C_0$ do \\
    \> if $C_1$ \= \\
     \> then \\
     \> \>   $X := A_1(X)$; \\
    \> elif $C_2$ \\
    \> then \\
    \> \>    $X := A_2(X)$; \\
      \>   $\cdots$ \\
    \> elif $C_m$ \\
       \> then \\ 
       \> \> $X := A_m(X)$; \\
       \> end if\\
end do
\end{tabbing}
}}
}
\end{center}
where 
\begin{enumerate}
\item $X=x_1,x_2,\ldots,x_s$ is a list of $s$ scalar loop variables,
taking values from $\mathbb{Q}$; % involved in  $\mathcal{L}$. 
\item the initial values of the loop are constrained by polynomial equations and polynomial inequations;
   % by  
   % a list of polynomial equations and polynomial inequations; %in $X$ over $\mathbb{Q}$
% form a non-empty constructible set in $\mathbb{Q}^s$;
% called the loop initial value set. 
% Inside the loop, there are 
% $m$ exclusive execution branches.
% In each execution branch, there is only an 
% simultaneous assignment statement. % and invertible.
% We assume that
\item the $C_i$'s are pairwise exclusive
algebraic conditions (polynomial equations and polynomial inequations)
% quantifier free formulae %conditions on the 
on $X$; 
\item the $A_i$'s are polynomial functions of $X$
with coefficients from $\mathbb{Q}$.
\end{enumerate}

In our loop model, when the loop 
body contains assignments only (thus no branches),
the assignment indeed induces a recurrence relation among the loop variables,
which are viewed as recurrence variables. In this case,
we shall simply refer to the loop as this recurrence
relation and the initial values of the loop. 
Here we will show briefly how
the loop invariant can be computed by explicitly solving 
the recurrence relation. Later, in our theoretical analysis, presented 
in Sections~\ref{sec:algebraicRelations},
we will focus on the study degree and dimension of the invariant ideal 
of such kind of loops, where the induced recurrence relation is so-called
$P$-solvable recurrence.

\ifCoRR
\begin{Example}
\label{ex: intro2}
Consider the loop computing the sequence of the Fibonacci numbers:
   \begin{center}
    \fbox{
\parbox{0.1\textwidth}{
    {\bf
\begin{tabbing}
    % $ (x,y) \in { (x, y) | x^2-y=0, x>y }$;
    $y := 1$; \\ % initialize
    $ x:=0$;\\
    while \=  true  do \\  % equiv to 0 < x < 1
     \> $(x,y) := (y, x+y)$; \\
    end while \\
\end{tabbing}
}}}
\end{center}
Viewing $(x,y)$ as two recurrences variables, the loop is actually
computing the two recurrence sequences of values of $x$ and $y$
defined by the following recurrence relation and initial condition:

% the $(x,y)$-recurrence 
$$x(n+1)=y(n), y(n+1) = x(n) + y(n), \mbox{ with } x(0)=0,y(0)=1.$$
% where $x,y$ denote two consecutive Fibonacci numbers and
% with $x(0)=0,y(0)=1$. Note that
% $x(n)$ and $y(n)$ are two consecutive
% elements of in the Fibonacci sequence. 
 We can write down the closed form for $x(n)$ and $y(n)$ as follows:
\begin{equation*}
\begin{array}{rcl}
x(n) &=& \frac{(\frac{\sqrt{5}+1}{2})^n}{\sqrt{5}} -\frac{(\frac{-\sqrt{5}+1}{2})^n}{\sqrt{5}}, \\
y(n) &=& \frac{\sqrt{5}+1}{2} \, \frac{(\frac{\sqrt{5}+1}{2})^n}{\sqrt{5}} -\frac{-\sqrt{5}+1}{2}\, \frac{(\frac{-\sqrt{5}+1}{2})^n}{\sqrt{5}}.
\end{array}
\end{equation*}
Let $a, u, v$ be $3$ variables.
Replace $(\frac{\sqrt{5}+1}{2})^n$ 
(resp.  $(\frac{-\sqrt{5}+1}{2})^n$)  
by $u$ (resp. by $v$); 
replace $\sqrt{5}$ by $a$. 
Taking the dependencies
$u^2 \, v^2 =1, a^2=5$ on the new variables into account,  
% Denote by $J$ the ideal $\langle x-a u/5+a v/5, y-a (a/2+1/2) u/5+a (-a/2+1/2) v/5, a^2-5, u v-1 \rangle  \; \cap \; \mathbb{Q}[x,y],$
% which turns out to be $\langle 1+x^2+y \, x-y^2 \rangle$.
% denote by $J$
 the invariant ideal of the loop is 
%(the elimination ideal) 
$$\langle x-\frac{a u}{5}+\frac{a v}{5}, y-a \frac{a+1}{2} \frac{u}{5}+a \frac{-a+1}{2} \frac{v}{5}, a^2-5, u^2 v^2-1 \rangle  \; \cap \; \mathbb{Q}[x,y],$$
which turns out to be $\langle 1-y^4+2 x y^3+x^2 y^2-2 x^3 y-x^4 \rangle$.
%
% One might deduce 
% that the invariant ideal of $(x,y)$ is  $J$. However, this is not true,
% the invariant ideal of $(x,y)$ is a proper subset ideal of $J$.
% This is justify by the fact
%  that $(x(1), y(1)) = (1,1)$ is not in $Z(J)$; indeed,
% for each $k=0\cdots \infty$, we have $(x(2\,k), y(2\,k))  \in Z(J)$ and 
% $(x(2\,k+1), y(2\,k+1)) \not \in Z(J)$.
% Suppose we would like to compute polynomial $f(x,y)$ s.t. 
% $f(x(n), y(n))=0$ holds for all $n \in \mathbb{N}$. 
%
% One direct way to solve this is try to compute the explicit forms for 
\end{Example}
\fi

\ifCoRR
\subsection{Poly-geometric summation}

As we discussed in the previous subsection, 
the study of loops without branches can be reduced to 
the study of recurrence sequences. 
In this subsection, we recall several well-known 
notions together with related 
results adapted to our needs.

Those notions and results can usually be stated in a 
more general context, e.g. the notion
of multiplicative relation can be defined
among elements of an arbitrary Abelian group,
whereas we define it for a multiplicative group of algebraic numbers.
\else 
We recall now several well-known 
notions on recurrence relations together with related 
results adapted to our needs.
\fi

\begin{Definition}
    \label{def: poly-geo}
Let $\alpha_1, \ldots, \alpha_k$  
be $k$ elements of $\overline{\mathbb{Q}}^* \setminus \{1\}$.
Let $n$ be a variable taking non-negative integer values. 
We regard $n, \alpha_1^n, \ldots,\alpha_k^n$ 
as independent variables and we call $\alpha_1^n, \ldots,\alpha_k^n$
$n$-exponential variables.
% , or in other words, as pairwise distinct symbols.
Any polynomial of 
$\overline{\mathbb{Q}}[n, \alpha_1^n, \ldots,\alpha_k^n]$
is called 
a {\em poly-geometrical expression in} $n$ {\em over} $\overline{\mathbb{Q}}$
w.r.t. 
$\alpha_1, \ldots, \alpha_k$. 

Let $f,g$ be two poly-geometrical expressions  $n$ over $\overline{\mathbb{Q}}$
w.r.t.  $\alpha_1, \ldots, \alpha_k$.
Given a non-negative integer number $i$, 
we denote by  $f|_{n=i}$ the {\em evaluation} of $f$ at $i$,
which is obtained by 
substituting all occurrences of $n$ by $i$ in $f$.
We say that $f$ and $g$ are {\em equal} whenever $f|_{n=i} = g|_{n=i}$ 
holds for all  non-negative integer $i$.

We say that $f(n)$ is in {\em canonical form}
if there exist
\begin{enumerate}[$(i)$] 
      \item finitely many numbers 
             $c_1, \ldots, c_m \in \overline{\mathbb{Q}}^*$, and
      \item finitely many pairwise different couples
            $(\beta_1, e_{1}),  \ldots, (\beta_m, e_{m})$
            all in 
        $(\overline{\mathbb{Q}}^* \setminus \{1\}) \times \mathbb{Z}_{\geq 0}$, and
     \item a polynomial $c_0(n) \in {\Q}[n]$,
\end{enumerate}
such that each $\beta_1, \ldots, \beta_m$ is a product of some 
of the  $\alpha_1, \ldots, \alpha_k$  and such that 
the poly-geometrical expressions $f(n)$ 
and $\sum_{i=1}^m\,c_i \, \beta_i^n\,n^{e_{i}}\; +\; c_0(n)$ are equal.
When this holds, the polynomial $c_0(n)$ is called the {\em exponential-free part}
of $f(n)$.
\end{Definition}

\begin{Remark}
Note that sometime when referring to poly-geometrical expressions, 
for simplicity, we allow $n$-exponential terms with base $0$ or $1$, 
that is, terms with $0^n$ or $1^n$ as factors.
Such terms will always be evaluated to $0$ or $1$ respectively.
\end{Remark}

Proving the following result is routine.

\begin{Lemma}
With the notations of Definition~\ref{def: poly-geo}.
Let $f$ a poly-geometrical expression in  $n$ over $\overline{\mathbb{Q}}$
w.r.t.  $\alpha_1, \ldots, \alpha_k$.
There exists a unique poly-geometrical expression $c$ in  $n$ over $\overline{\mathbb{Q}}$
w.r.t.  $\alpha_1, \ldots, \alpha_k$ such that 
$c$ is in canonical form and such that $f$ and $c$ are equal.
We call $c$ the {\em canonical form} of $f$.
\end{Lemma}

\ifCoRR
\begin{Example}
The closed form $f := \frac{(n+1)^2\,n^2}{4} $ of $\sum_{i=0} ^n i^3$
is a poly-geometrical expression in $n$ over $\overline{\mathbb{Q}}$ without
$n$-exponential variables.  
The expression $g := n^2 \, 2^{(n+1)}-n \, 2^n \, 3^{\frac{n}{2}}$
is a poly-geometrical in $n$ over $\overline{\mathbb{Q}}$ w.r.t. $2,3$.
Some evaluations are: $f|_{(n=0)}=0, f|_{n=1}=1, g|_{n=0}=0, 
g|_{n=2}= 8 $.  
% w.r.t. $2,3$.
%Consider now a fixed non-negative integer $k$.
%          The sum $\sum_{i=1}^{n-1} i^k$ has $n-1$ terms
%           while its closed form~\cite{GG99} below 
%        $$\sum_{i=1}^k \; {k \brace i} \, \frac{ n^{\underline{i+1}}}{ i+1}$$
%  has a fixed number of terms and thus 
%        is poly-geometrical in $n$ over $\overline{\mathbb{Q}}$.
\end{Example}

\begin{Notation}
%% We call an {\em arihmetic expression}
%% any multivariate polynomial with coefficients in $\overline{\mathbb{Q}}$.
Let $x$ be an arithmetic expression and let $k \in \mathbb{N}$.
Following \cite{GG99},
we call \emph{$k$-th falling factorial of} $x$
and denote by $x^{\underline{k}}$ the 
product $$x \, (x-1) \, \cdots (x-k+1).$$
We define $x^ {\underline{0}} := 0$.
For $i=1,\ldots, k$, we denote by ${k \brace i}$ the number of ways 
to partition
$k$ into $i$ non-zero summands, that is, the \emph{Stirling number of the second kind}
also denoted by $S(n,k)$.
We define ${k \brace 0} := 0$. 
Finally, we shall make use of the convention $0^0=1.$
\end{Notation}

\begin{Example}
The expression $n^2 \, 2^{(n+1)}-n \, 2^n \, 3^{(n/2)}$
is clearly poly-geometrical in $n$ over $\overline{\mathbb{Q}}$.
Consider now a fixed non-negative integer $k$.
          The sum $\sum_{i=1}^{n-1} i^k$ has $n-1$ terms
           while its closed form~\cite{GG99} below 
        $$\sum_{i=1}^k \; {k \brace i} \, \frac{ n^{\underline{i+1}}}{ i+1}$$
  has a fixed number of terms and thus 
        is poly-geometrical in $n$ over $\overline{\mathbb{Q}}$.
\end{Example}

The following result is proved in~\cite{GG99}.
\begin{Lemma}
    \label{lem: ff}
    Let $x$ be an algebraic expression and let $k \in \mathbb{N}$.
    Then we have 
    $$x^k = \sum_{i=1}^k \, {k \brace i } \, x^{\underline{i}}.$$
\end{Lemma}

\begin{Notation}
Let $r \in  \overline{\mathbb{Q}}$ and let $k$ be a non-negative integer. 
We denote by $H(r,k,n)$ the following symbolic summation 
$$H(r,k,n) :=  \sum_{i=0}^{n-1}\;  r^i \; i^{\underline{k}}.$$
\end{Notation}

One can easily check that $H(r,0,n) = \frac{r^n-1}{r-1}$ holds
for $r \neq 1$.
Moreover, we have the following result.

\begin{Lemma}
    \label{lem: poly-geo-sum}
	Assume $r \neq 0$. Then, we have
\begin{equation}
\label{eq: poly-geo-sum}
( r -1 )\,H(r,n,k) \; =  \;  (n-1)^{\underline{k}} \, r^{n} - r \, k \, \, 
H(r, k-1,n-1).
\end{equation}
Moreover, we have
\begin{enumerate}
    \item[$(i)$]
if $r=1$, then $H(r,n,k)$ equals to  $\frac{n^{\underline{k+1}}}{k+1}$, 
which is a polynomial in $n$ over $\overline{\mathbb{Q}}$ 
of degree $k+1$.

    \item[$(ii)$] if $r \neq 1$, then $H(r,n,k)$ has a closed form
        like $r^n \, f(n) + c$, where $f(n)$ is a polynomial 
        in $n$ over $\overline{\mathbb{Q}}$ 
of degree $k$ and $c$ is a constant in $\overline{\mathbb{Q}}$.
\end{enumerate}
\end{Lemma}

\begin{proof} %[of Lemma\ref{lem: poly-geo-sum}]
We can verify Relation~(\ref{eq: poly-geo-sum})
by expanding $H(r,n,k)$ and $H(r, k-1,n-1)$.
Now let us show the rest of the conclusion.
First, assume $r = 1$. 
With Relation~(\ref{eq: poly-geo-sum}), we have 
$$k \, H(r, k-1,n-1) = (n-1)^{\underline{k}}.$$
Therefore, we deduce 
$$ H(r,n,k) = \frac{n^{\underline{k+1}}}{k+1}.$$
One can easily check that $\frac{n^{\underline{k+1}}}{k+1}$ 
is a polynomial in $n$ over $\overline{\mathbb{Q}}$ 
and $\deg(s, n) = k+1$.

From now on assume $r \neq 1$.
We proceed by induction on $k$.
When $k=0$, we have $H(r,0,n) = \frac{r^n-1}{r-1}$.
We rewrite $\frac{r^n-1}{r-1}$ as $$r^n\, \frac{1}{r-1} -\frac{1}{r-1},$$
which is such a closed form.
Assume there exists a closed form $r^n\, f_{k-1}(n) + c_{k-1}$ for $H(r,k-1,n)$,
where $f_{k-1}(n)$ is a polynomial in $n$ over $\overline{\mathbb{Q}}$ 
of degree $k-1$.
Define  
$$s := 
\frac{(n-1)^{\underline{k}} \, r^{n} - r \, k \, ( r^{n-1}\, f_{k-1}(n-1) + c_{k-1} )}{r-1}
.$$
It is easy to verify that $s$ is a closed form of $H(r,n,k)$.
We rewrite $s$ as 
$$ r^{n} \, \frac{(n-1)^{\underline{k}} - k\, f_{k-1}(n-1) }{r-1} - \frac{ r\,k\,c_{k-1}} {r-1},$$
and one can check the later form satisfies the requirements of $(ii)$ in the conclusion.
This completes the proof.
\end{proof}

\begin{Lemma}
    \label{lem: poly-geo}
Let $k \in \mathbb{N}$ and let $\lambda$ be a non zero algebraic number 
over $\mathbb{Q}$. 
Consider the symbolic summation  $$S := \sum_{i=1}^n\;  i^k \, \lambda^i.$$
\begin{enumerate}
    \item if $\lambda = 1$, then there exists a closed form $s(n)$
    for $S$, where $s$ is a polynomial in $n$ 
		  over $\overline{\mathbb{Q}}$ of degree $k+1$.
    \item if $\lambda \neq 1$, then there exists a closed form $\lambda^n\,s(n) + c$
    for $S$, where $s$ is a polynomial in $n$ 
		  over $\overline{\mathbb{Q}}$ of degree $k$ and $c \in \overline{\mathbb{Q}}$ is a constant.
  \end{enumerate}
\end{Lemma}

\begin{proof}%[of Lemma~\ref{lem: poly-geo}]
By Lemma~\ref{lem: ff},
 we deduce 
 \begin{equation*}
 \begin{array} {cl}
 \sum_{i=1}^n\;  i^k \, \lambda^i  \! & =  \;  \sum_{i=1}^n \; \left( 
  \sum_{j=1}^k \, {k \brace j } \, i^{\underline{j}} \right)
 \, \lambda^i\\
 \! & = \; \sum_{j=1}^k \, \left({k \brace j } \, 
\sum_{i=1}^n \; 
   i^{\underline{j}} \, \lambda^i\right) \\
 \! & = \; \sum_{j=1}^k \, \left({k \brace j } \, H(\lambda,j, n)  \right)
 \end{array}
 \end{equation*}
Then, the conclusion follows from Lemma~\ref{lem: poly-geo-sum}.
\end{proof}

The following definition of multiplicative relation
specializes the general definition of multiplicative relation
to non-zero algebraic numbers.

\begin{Definition}[Multiplicative relation]
\label{defi:multiplicationRelation}
Let $k$ be a positive integer.
Let $A := (\alpha_1, \ldots, \alpha_k)$ be a sequence of $k$ 
non-zero algebraic numbers over $\mathbb{Q}$
and $\mathbf{e} := (e_1, \ldots, e_k)$ be a sequence of $k$
integers.
We say that $\mathbf{e}$ is a \emph{multiplicative relation on} $A$
if $\prod_{i=1}^k \, \alpha_i^{ e_i } = 1$ holds.
Such a multiplicative relation is said \emph{non-trivial} if
there exists $i \in \{ 1, \ldots, n \}$ such that $e_i \neq 0$ holds.
If there exists a non-trivial multiplicative relation on $A$, 
then we say that $A$ is \emph{multiplicatively dependent};
otherwise,  we say that $A$ is \emph{multiplicatively independent}.
\end{Definition}

All multiplicative relations of $A$ form a lattice,
called the \emph{multiplicative relation lattice} on $A$,
which can effectively be computed, for instance with the
algorithm proposed by G. Ge in his PhD thesis~\cite{Ge93}.
\fi

\ifCoRR
For simplicity, we need the following generalized notion of multiplicative relation ideal,
which is defined for a sequence of algebraic numbers that may contain $0$ and repeat elements. 
\fi

\begin{Definition}
\label{defn: multiplication-ideal}
Let $A := (\alpha_1, \ldots, \alpha_k)$ be a sequence of $k$ algebraic 
numbers over $\mathbb{Q}$.
Assume w.l.o.g. that there exists an index ${\ell}$, 
with $1 \leq {\ell} \leq k$,  such that
$\alpha_1, \ldots, \alpha_{\ell}$ are non-zero
and $\alpha_{{\ell}+1}, \ldots, \alpha_k$
are all zero.
We associate each  $\alpha_i$ with a variable $y_i$,
where $y_1,\ldots, y_k$ are different from each other.
We call the \emph{multiplicative relation ideal of} $A$ 
{\em associated with variables}
$y_1,\ldots, y_k$, the  binomial ideal of $\mathbb{Q}[y_1,y_2,\ldots,y_k]$
generated by 
	$$\{\prod_{j \in \{ 1, \ldots, {\ell}\}, \, v_j >0} y_j^{v_j} - \prod_{i \in \{ 1, \ldots, {\ell}\}, \, v_i <0} y_i^{-v_i} \ \mid \ (v_1, \ldots, v_{\ell}) \in Z \}$$ 
and $\{y_{{\ell}+1}, \ldots, y_k\}$,
 denoted by \MI{A;y_1, \ldots,y_k},
where $Z$ is the multiplicative relation 
\ifCoRR
lattice 
\else
lattice (see for a study of this notion cite~\cite{Ge93})
\fi
on $(\alpha_1, \ldots, \alpha_{\ell}).$ %%
When no confusion is possible, we shall omit 
writ-ting down the associated variables $y_1,\ldots, y_k$.
\end{Definition}

% Proving the following is routine. \fbox{Find a reference.}

\begin{Lemma}
\label{lem: transcent}
Let $\alpha_1, \ldots, \alpha_k$  
be $k$ multiplicatively independent elements of $\overline{\mathbb{Q}}$
and let $n$ be a non-negative integer variable.
Let $f(n)$ be a poly-geometrical expression in $n$ 
w.r.t.  $\alpha_1, \ldots, \alpha_k$.
Assume that $f|_{(n=i)}=0$ holds for all $i \in \mathbb{N}$.
Then, $f$ is the zero polynomial of 
$\overline{\mathbb{Q}}[n, \alpha_1^n, \ldots,\alpha_k^n]$.
\end{Lemma}

\ifCoRR
The following definition will be convenient in later statements.
\fi

\begin{Definition}[Weakly multiplicative independence] 
Let $A := (\alpha_1, \ldots, \alpha_k)$ be a sequence of $k$ 
non-zero algebraic numbers over $\mathbb{Q}$ and 
let $\beta \in \overline{\mathbb{Q}}$.
We say $\beta$ is \emph{{\wm} independent} w.r.t. $A$, if 
there exist no non-negative integers $e_1, e_2, \ldots, e_k$
such that 
 $\beta = \prod_{i=1}^k \, \alpha_1^{ e_i }$ holds.
Furthermore, we say that 
$A$ is \emph{{\wm} independent} if
\begin{enumerate}
\item[$(i)$] $\alpha_1 \neq 1$ holds, and
\item[$(ii)$] $ \alpha_i$ is {\wm} independent w.r.t.\\
$\{\alpha_1,\ldots,\alpha_{i-1}, 1\} $, for all $i=2,\ldots,s$. 
\end{enumerate}
\end{Definition}

It is not hard to prove the following lemma on the shape of 
closed form solutions of single-variable linear recurrences
involving poly-geometrical expressions. 
\ifCoRR
% See the appendix for a proof.
For the proof, we need the following lemma,  which is easy to check, see for instance~\cite{Osb00}.
\begin{Lemma}
\label{lem: lde-uni}
	Let $n$ a variable holding non-negative integer values.
    Let $a$ and $b$ be two sequences in $\mathbb{Q}$ indexed by $n$.
    Consider the following recurrence equation of variable $x$: 
    $$x(n)=a(n-1)\, x(n-1) + b(n-1).$$
    Then we have
    $$ x(n) = \prod_{i=0}^{n-1}\,a(i) \;  \left(x(0) + \sum_{j=0}^{n-1} \, \frac{b(j)}{\prod_{s=0}^{j} \,a(s) } \right).$$
\end{Lemma} 
\fi

\begin{Lemma}
\label{lem: solvable-base}
% Let $\mathbb{F}$ be a field.
% Denote by $\mathbb{F}^{*}$ the group of the units of $\mathbb{F}$,
% that is, $\mathbb{F}^{*} = \mathbb{F} \setminus \{ 0\}$.
Let $\alpha_1, \ldots, \alpha_k$ be 
% a set of multiplicatively independent
% algebraic number over $\mathbb{Q}$.
$k$ elements in  $\overline{\mathbb{Q}}^* \setminus \{1\}$.
Let $\lambda \in \overline{\mathbb{Q}}^*$ .
% such that $\lambda$ is {\wm} independent (thus not $1$) 
% w.r.t. $\{\alpha_1, \ldots, \alpha_k\}$.
    Let  $h(n)$ be a poly-geometrical expression in $n$ over $\overline{\mathbb{Q}}$
    w.r.t. $\alpha_1, \ldots, \alpha_k$.
Consider the following single-variable recurrence relation $R$:
$$x(n+1)=\lambda x(n) + h(n).$$ 
Then, there exists a poly-geometrical expression $s(n)$ in $n$ 
over  $\overline{\mathbb{Q}}$   w.r.t. $\alpha_1, \ldots, \alpha_k$ 
such that we have 
      $$
\deg(s(n),\alpha_i^n) \leq \deg(h(n),\alpha_i^n) \ \ {\rm and} \ \ 
\deg(s(n),n)  \leq \deg(h(n),n)+1,
  $$
and such that
\begin{itemize}
\item if $\lambda = 1$ holds,  then $s(n)$ solves $R$,
\item if $\lambda \neq 1$ holds, then there exists
 a constant $c$ depending on $x(0)$ 
  (that is, the initial value of $x$) such that
  $c\, \lambda^n + s(n)$ solves $R$.
\end{itemize}
%%
% with initial value $x(0)$.
 %    ny solution   
 % $x(n)$% solution of a $P$-solvable recurrence 
 %  is a poly-geometrical expression in $n$ over $\mathbb{Q}$  w.r.t.    
 %   $\alpha_1, \ldots, \alpha_k, \lambda$.
    % is $poly-rank(x(n)) \leq poly-rank(h(n))+1$, 
 %    Moreover, % we have
 %%
%  moreover,
%  if the constant term of the compression of $h(n)$ is $0$ when viewed as a polynomial of 
%  the $n$-exponential variables, then 
%  we can further require that $\deg(s(n),n)  \leq \deg(h(n),n)$.
%%
%%
   % moreover, if $\lambda$ is {\wm} independent w.r.t. $\alpha_1, \ldots, \alpha_k$, we can  choose $s(n)$  such that $\deg(s(n),n) = \deg(h(n),n)$ holds.
%%  rsolve({x(n+1)=2*x(n) + 2^n, x(0)=1}, x);
%%                               n
%%                              2                  n
%%                             ---- + (n/2 + 1/2) 2
%%                              2
%%
%%
 Moreover, in both cases, if the exponential-free part of the 
canonical form of $(\frac{1}{\lambda})^n\,h(n)$ 
  is $0$, then 
  we can further require that $\deg(s(n),n)  \leq \deg(h(n),n)$ holds.
\end{Lemma}
\ifCoRR
\begin{proof}
% [{\bf Lemma~\ref{lem:  solvable-base}}]
% If $\lambda =0$, then the conclusion is trivially true.
% From now on we assume that $\lambda  \neq 0$ holds.
%% 
By Lemma~\ref{lem: lde-uni},
we have 
\begin{equation}
\label{eq: 2.0}
x(n) = \lambda^n \left(x(0) +  \sum_{j=0}^{n-1} \, \frac{h(j)}{\lambda^{j+1}} \right).
\end{equation}

Denote by ${\tt terms}(h)$ all the terms of the canonical form of $h(n)$.
% in $\overline{\mathbb{Q}}[n, \alpha_1^n, \ldots, \alpha_k^n]$.
Assume each $t \in {\tt terms}(h)$ is of form $$c_t \, n^{q_t}\,\beta_t^n,$$
 where 
	$c_t$ is a constant in $\overline{\mathbb{Q}}$, $q_t$ is a non-negative integer
	and  $\beta_t$ is a product of finitely many elements 
     (with possible repetitions) from $\{ \alpha_1, \ldots, \alpha_k\}$.
Define $g(n) := \frac{h(n)}{\lambda^{n+1}}$. 
Then $g(n)$ is a poly-geometrical expression in $n$
w.r.t. $\{\beta_t\}_{t \in {\tt terms}(h)}, \frac{1}{\lambda}$.
% w.r.t. $\alpha_1, \ldots, \alpha_k, \frac{1}{\lambda}$.
% Denote by ${\tt terms}(g(n))$ all the terms of $g(n)$
% as a polynomial 
% in $\overline{\mathbb{Q}}[n, \alpha_1^n, \ldots, \alpha_k^n, \frac{1}{\lambda^n}]$.
Clearly we have 
$$g(n) = \sum_{t \in  {\tt terms}(h(n))} \, \frac{c_t}{\lambda}\, n^{q_t} \,(\frac{\beta_t}{\lambda})^n.$$
% Let  $t = c_t \, n^{q_t}\,\beta_t^n \in  {\tt terms}(h(n))$
% Then each 
% term $t(n) \in {\tt terms}(g(n))$ is of the form
%     $c_t \, n^{q_t}\, (\frac{\beta_t}{\lambda})^n$, where 
%	$c_t$ is a constant in $\overline{\mathbb{Q}}$,  
%        $q_t$ is a non-negative integer
%	and $\beta$ is a product of finitely many elements 
%     (with possible repetitions) from $\{ \alpha_1, \ldots, \alpha_k\}$.

Therefore, we have 
\begin{equation}
% \label{eq: 2.1}
\sum_{j=0}^{n-1} \, \frac{h(j)}{\lambda^{j+1}} =
\sum_{t \in {\tt terms}(h)} \sum_{j=0}^{n-1} \,   \frac{c_t}{\lambda}\, j^{q_t} \,(\frac{\beta_t}{\lambda})^j.
\end{equation}

According to  Lemma~\ref{lem: poly-geo},
for each $t \in {\tt terms}(h)$, 
% c_t \, n^q_t\, (\frac{\beta_t}{\lambda})^n
we can find  a poly-geometrical expression $$s_t :=  (\frac{\beta_t} {\lambda})^n f_t(n) + a_t$$
in $n$  over $\overline{\mathbb{Q}}$ w.r.t. $\frac{\beta_t}{\lambda}$ satisfying
\begin{enumerate}
	\item $s_t = \sum_{j=0}^{n-1} \,  \frac{c_t}{\lambda}\, j^{q_t} \,(\frac{\beta_t}{\lambda})^j$;
	\item $f_t$ is a polynomial in $n$ over $\overline{\mathbb{Q}}$ of degree
$q_t$ ( if $\beta_t \neq  \lambda$) or $q_t +1$ (if $\beta_t = \lambda$), 
and $a_t$ is a constant in $\overline{\mathbb{Q}}$; note in the later case,
$c_t \, n^{q_t} \,(\frac{\beta_t}{\lambda})^n$
is a summand of the constant term of the canonical form of $(\frac{1}{\lambda})^n\,h(n)$
  is $0$ when viewed as a polynomial of 
  the $n$-exponential variables.
\end{enumerate}
Therefore,  using $s_t$ ($ \forall t \in {\tt terms}(h)$), we can simplify the right hand side of Equation (\ref{eq: 2.0}) to
\begin{equation}
\label{eq: 2.2}
\left(x(0)+\sum_{ t \in {\tt terms}(h)}\, a_t\right) \, \lambda^n + \sum_{ t \in {\tt terms}(h)}\,f_t(n) \, \beta_t^n.
\end{equation}

% Note 
% each $s_t^*$ in $\alpha_1, \ldots, \alpha_k, \lambda$
% s.t.  $\deg(s_t^*(n),\alpha_i^n) = \deg(t(n),\alpha_i^n)$, $\deg(s_t^*(n),\lambda^n) \leq 1$.
% where $t(n)$ is viewed as a monomial in  $\overline{\mathbb{Q}}[n, \alpha_1^n, \ldots, \alpha_k^n]$.
Assume for each $t \in {\tt terms}(h)$, we have $\beta_t = \alpha_1^{e_{t,1}} \, \alpha_1^{e_{t,2}}  \,\cdots\, \alpha_1^{e_{t,k}}$. 
% Then substitute each $\beta_t^n$ in  \label{eq: 2.2}

Define $$\beta_t(n) := (\alpha_1^n)^{e_{t,1}} \, (\alpha_1^n)^{e_{t,2}}  \,\cdots\, (\alpha_1^n)^{e_{t,k}},$$
% \deg(\beta_t(n),\alpha_i^n)
 $$c := x(0)+\sum_{ t \in {\tt terms}(h)}\, a_t
\  \mbox{ and } \ s(n) := \sum_{ t \in {\tt terms}(h)}\,f_t(n) \, \beta_t(n).$$

% Note that $\deg(\beta_t(n),\alpha_i^n)$ equals to the power of 
% $\alpha_i$ in $\beta_t$ and 
It is easy to deduce
$\deg(s(n), \alpha_i^n) = \max_{t \in  {\tt terms}(h)}(\deg(\beta_t(n),\alpha_i^n) \leq \deg(h(n), \alpha_i^n)$.
% \fbox{Complete the proof of the degree inequalities.}
%\fbox{To this end, make $t$ the generic term of $h$.}
Finally, one can easily verify that $c$ and $s(n)$ satisfy the requirements in the conclusion.
\end{proof}
\fi

\ifCoRR
\begin{Remark}
In Lemma~\ref{lem: solvable-base}, 
if $\lambda$ is {\wm} independent w.r.t. $\alpha_1, \ldots, \alpha_k$, then we 
know that the exponential-free part of the canonical form of $(\frac{1}{\lambda})^n\,h(n)$
  is $0$, without computing
  the canonical form explicitly.
\end{Remark}
\fi
%The case of $0$ eigenvalues is not clear.
%> rsolve({x(n+1)=x(n)+a*b^n, x(0)=c},x);
%                                              n
%                                     a     a b
%                               c - ----- + -----
%                                   b - 1   b - 1

\ifheight
\begin{Notation}
Let $r := \frac{q}{p} \in \mathbb{Q}$ 
where $p, q \in \mathbb{Z}$ are coprime.
the height of $r$, denoted by $\hbar(r)$,
is defined as 
$\max(\log_2{|p|}, \log_2{|q|})$.

Let $\alpha$ be an algebraic number over $\mathbb{Q}$
with minimal polynomial in $\mathbb{Q}[x]$ as 
$$\sum_{i=0}^n a_i \; x^i .$$
Then the \emph{height} of $\alpha$, denoted by $\hbar(\alpha)$, is defined to 
be $n + \sum_{i=0}^n \, \hbar(a_i)$.
% We can not take any zero polynomial to estimate the size of a_i
% > factor(x^5-2^5);
%                               4      3      2
%                     (x - 2) (x  + 2 x  + 4 x  + 8 x + 16)
% however, we can estimate the height of irreducible factors

Let $\alpha_1, \alpha_2, \ldots, \alpha_k$ be $k$ algebraic numbers over $\mathbb{Q}$. Consider a polynomial $p \in Q(\alpha_1, \alpha_2, \ldots, \alpha_k)[x_1, x_2, \ldots, x_s]$.
Let $c_1,c_2, \ldots,c_{\ell}$ be the coefficients of $p$.
Then the \emph{height} of $p$, denoted by $\hbar(p)$, is defined to 
be $\max(\hbar(c_i))$.
\end{Notation}

The following results is estimated by the size of determinants, 
see~\cite{LMW06}.
\begin{Lemma}
	Let $M_{n \times n}$ be a matrix in $\mathbb{Q}^{ n \times n}$,
	whose coefficients have heights upper bounded by $\hbar$.
	Then the coefficient size of the characteristic polynomial of 
	$M$ is upper bounded by $n (\hbar+\log_2(n))$.
\end{Lemma}
\fi

\ifCoRR

\subsection{Degree preliminaries}
In this subsection, we review some notions and results on the 
degree of algebraic varieties. 
Up to our knowledge, 
Proposition~\ref{prop:  bi-rational-degree} 
is a new result which provides a degree estimate for an
ideal of a special shape and which can be applied to degree estimate
of loop invariant ideals.
Throughout this subsection, let $\mathbb{K}$ be an algebraically closed field. 
Let $F$ be set of polynomials of $\mathbb{K}[x_1,x_2,\ldots,x_s]$.
We denote by $V_{\mathbb{K}^s}(F)$  (or simply by $V(F)$ when
no confusion is possible) the zero set of the ideal generated by 
$F \subset \mathbb{K}[x_1,x_2,\ldots,x_s]$ in $\mathbb{K}^s$.
% For V embedded in a projective space Pn and defined 
% over some algebraically closed field K

%\begin{Notation}
%Let $I$ be an ideal of $\mathbb{K}[x_1,x_2,\ldots,x_s]$. 
%We denote by $V_{\mathbb{K}^s}(I)$  (or simply $V(I)$ when
%there is not confusion possible) the zero set of $I$ in $\mathbb{K}^s$.
%Denote by $\dim(I)$ the dimension of $I$.
%% (which equals the dimension of 
%% $V(I)$ in  $\overline{\mathbb{Q}}^s$). 
%
%Let $C$ be a polynomial system or an algebraic condition of $x_1,x_2,\ldots,x_s$.
%Then we denote by $Z(C)$ the set of solutions of $C$ in 
%$\overline{\mathbb{Q}}^s$, which is a \emph{constructible set} in $\overline{\mathbb{Q}}^s$.
%
%% Let $C := [F, H]$ be an algebraic system on the variables $x_1,x_2,\ldots,x_s$
%% over $\overline{\mathbb{Q}}$.
% 
%For a variety $W$ in $\overline{\mathbb{Q}}^s$, we denote by $\deg(W)$
%the degree of the $W$, see~\cite{Eisenbud1995} for a definition.
%For the ideal $I \subset \overline{\mathbb{Q}}[x_1,x_2,\ldots,x_s]$,
%we define its degree, denoted by $\deg(I)$, to
%be $\deg(V(I))$.
%% Denote by $\deg(V(I))$ the degree of $I$ (which equals
%% the degree of $V(I)$. 
%\end{Notation}
% extend to algebraic recurrence.z	

\begin{Definition}
Let $V \subset \mathbb{K}^s$ be an $r$-dimensional 
equidimensional algebraic variety.
The number of points of intersection of $V$
with an $(n-r)$-dimensional generic linear subspace $L \subset \mathbb{K}^s$
is called the \emph{degree} of $V$~\cite{CLO97}, denoted by $\deg(V)$. 
The degree of a non-equidimensional variety is defined to be 
the sum of the degrees of its equidimensional components.
The degree of an ideal $I \subseteq \mathbb{K}[x_1,x_2,\ldots,x_s]$
is defined to be the degree of the variety of $I$ in $\mathbb{K}^s$.
\end{Definition}

We first  review a few well-known lemmas. 
Note that, for a zero-dimensional algebraic variety, 
the degree is just the number of points in that variety. 
\begin{Lemma}
\label{lem: degree-max}
Let $V \subset \mathbb{K}^s$ be an $r$-dimensional 
equidimensional algebraic variety  of degree $\delta$.
Let $L$ be an $(n-r)$-dimensional linear subspace.
Then, the intersection of $L$ and $V$ is either of positive 
dimensional or consists of no more than $\delta$ points.
\end{Lemma}

\begin{Lemma}
\label{lem: degree-lm}
Let $V \subset \mathbb{K}^s$ be a algebraic variety.
Let $L$ be a linear map from $\mathbb{K}^s$ to 
$\mathbb{K}^k$. Then we have $\deg(L(V)) \leq \deg(V)$.
\end{Lemma}

\begin{Lemma}[\cite{Heintz83}]
% Let $\tilde{\mathbb{Q}}/\mathbb{Q}$ be a finite algebraic
% field extension.
Let $I \subset \mathbb{Q}[x_1,x_2,\ldots,x_s]$ be a radical ideal of degree $\delta$.
Then there exist finitely many polynomials in $\mathbb{Q}[x_1,x_2,\ldots,x_s]$
generating $I$ and such that 
each of this polynomial has total degree less than or equal to $\delta$.
\end{Lemma}

\begin{Lemma}
\label{lem: degree-bezout}
Let $V := W \; \cap_{i=1}^e \, V_i$ with $\dim(W)=r$. Then we have 
$$ \deg(V) \leq \deg(W)\, \max (\{\deg(V_i)\mid i=1 \cdots e\})^r.$$
\end{Lemma}

\begin{Proposition}
\label{prop: bi-rational-degree}
Let $X = x_1,x_2,\ldots,x_s$ and $Y = y_1,y_2,\ldots,y_t$
be pairwise different $s+t$ variables.
Let $M$ be an ideal in $\mathbb{Q}[Y]$
of degree $d_M$ and dimension $r$. 
Let $f_1, f_2, \cdots, f_s$ be $s$ polynomials in $\mathbb{Q}[Y]$,
with maximum total degree $d_f$. 
Denote by $I$ the ideal $\langle x_1-f_1, x_2-f_2, \ldots, x_s-f_s\rangle$.
 Then the ideal $J :=  I+ M$ 
has degree upper bounded by $d_M \, {d_f}^r$.
\end{Proposition}

\begin{proof}
We assume first that $M$ is equidimensional.
Let $L := l_1,l_2,\ldots,l_r$ be $r$ linear forms in $X,Y$ such
that the intersection of the corresponding $r$  hyperplanes 
and $V(J)$ consists of finitely many points,
i.e. $H_{L} := J + \langle  L \rangle$ is zero dimensional.
By virtues of Lemma~\ref{lem: degree-max}, 
the degree of $J$ equals the maximal degree of $H_{L}$ 
among all possible choices of linear forms $l_1,l_2,\ldots,l_r$
satisfying the above conditions.

Let $L^*:=l_1^*,l_2^*,\ldots,l_r^*$,
where each $l_j^*$  ($j=1 \cdots r$)
is the polynomial obtained by substituting 
$x_i$ with $f_i$, for $i=1 \cdots s$, in the polynomials $l_j$.
Consider the ideal $L^*+M$ in $\mathbb{Q}[Y]$. 
It is easy to show that the canonical projection map 
$\Pi_Y$ onto the space of $Y$ coordinates  
is a  one-one-map between 
$V_{\mathbb{C}^t} ( M+ L^* )$ and
$\Pi_Y ( V_{\mathbb{C}^{t+s}}(H_{L}))$.
% as well as that the points in $V_{\mathbb{C}^t}( M+ L )$
%  and those in $V_{\mathbb{C}^{t+s}}(J)$ are in.
Therefore, $V_{\mathbb{C}^t}( M+ L^* )$ is zero dimensional
and $\deg(M+ L^*) = \deg(H_{L})$.
Hence, viewing 
$V_{\mathbb{C}^t}( M+ L^* )$ as 
$$V_{\mathbb{C}^t}(M)\bigcap_{j=1}^{r}\, V_{\mathbb{C}^t}(l_j^*)$$ 
and thanks to Lemma~\ref{lem: degree-bezout}, 
we have $\deg(V_{\mathbb{C}^t}( M+ L^* )) \leq d_M\, d_f^r$.
%  despite what the concrete
% forms of $L$ are. 
Therefore, we deduce that $\deg(J) = \max_{L} \deg(M+ L^*)  \leq d_M \, d_f^r$ holds, 
by Lemma~\ref{lem: degree-max}.

Assume now that $V_{\mathbb{C}^t}(M)$ is not necessarily equidimensional.
Let $V_1, V_2, \cdots, V_k$ be an irredundant equidimensional
decomposition of $V_{\mathbb{C}^t}(M)$, 
with corresponding radical ideals $P_1, P_2, \ldots, P_k$.
% of $M$.  
Then, applying the result proved in the first part of the proof to
each $I + P_i$ ($i=1\cdots k$),
we deduce 
$$
\begin{array}{r l}
\deg(J) = & \sum_{i=1}^k \; \deg(I+P_i)\\
        \leq & \sum_{i=1}^k\; \deg(P_i)\, d_f^{r_i}\\
	\leq & \sum_{i=1}^k\; \deg(P_i)\, d_f^r\\
         = &  d_M\, d_f^r,
\end{array}
$$
where $r_i$ is the dimension of $P_i$ in $\mathbb{Q}[Y]$.
This completes the proof.
\end{proof}

\begin{Remark}
For $J$ in Proposition~\ref{prop: bi-rational-degree}, a less tight degree bound
$$
 d_M\, d_f^{r+s}
$$
can easily be deduced from a generalized form of Bezout's bound, since $\deg(V_{\mathbb{C}^{t+s}}( M))$ has degree $d_M$ and if of dimension $r+s$ in
$\mathbb{C}^{t+s}$.
\end{Remark}

%\begin{Corollary}
%Using the same notations as in Proposition~\ref{prop: bi-rational-degree}, if $M$ is generated by polynomials 
%in $\mathbb{Q}[y_1,y_2, y_t]$ and the projection of .
%\end{Corollary}

%\begin{Remark}
%Moreover, the degree bound obtained in Proposition~\ref{prop: bi-rational-degree} is sharp, in the sense that, for any given ideal $M$ of degree $d_M$ and dimension $r$, for generic choice of polynomial in $y_1,y_2,\ldots, y_t$ in degree $d_f$. Then degree of $J$ can be $d_M \, d_f ^ r$.
%\end{Remark}

\begin{Example}
Consider $M := \langle n^2-m^3 \rangle$, $g_1 := x-n^2-n-m, g_2 := y-n^3-3 n+1$, and the ideal $J := M + \langle g_1, g_2 \rangle$. The ideal $M$ has degree $3$, and is of dimension $1$ in $\mathbb{Q}[n,m]$. The degree of $J$ is $9$, which can be obtained by computing the dimension of
% monomials in a standard basis of 
$$\mathbb{Q}(a,b,c,d,e)[x,y,m,n]/(J + \langle a\, x +b\, y +c \,  n +d\, m +e \rangle),$$ where $a,b,c,d,e$ are indeterminates.
% over the polynomial ring $\mathbb{Q}(a,b,c,d,e)[x,y,m,n]$
The degree bound estimated by Proposition~\ref{prop: bi-rational-degree}
is  $3 \times 3$, which agrees with the true degree. 
\end{Example} 
%> gb := Groebner:-Basis([f1,f2,f3,l], tdeg(x,y,m,n));
%> lms := map(Groebner:-LeadingMonomial, gb, tdeg(x,y,m,n));
%                                      2        2   3
%                          lms := [x, n , y n, y , m ]
%
%> sbasis := [1, n, m, m^2, y, y*m,y*m^2, n*m,n*m^2]

\else

In this paper, we make use of a few lemmas related
to the degrees of algebraic varieties.
These results can be found in the second version 
of~\cite{DBLP:journals/corr/abs-1201-5086},
in the {\em Computing Research Repository}.
\fi

% \input{algebraic_relation_recurrences}

% !Tex root=mx-2012-casc.tex
\section{Invariant ideal of $P$-solvable recurrences}
\label{sec:algebraicRelations}
In this section, we focus on loops 
with no branches, 
where the study of loop invariants of such loops reduces to 
% 
% are interested in so called $P$-solvable recurrence relations,
% in particular, 
the study of algebraic relations
among the recurrence variables. 
% By invariant ideal of 
In particular, we are interested in those
whose assignments induce  a called $P$-solvable recurrence.
We will first formalize the notion of $P$-solvable recurrence.
Then in the rest of this section, we will investigate
the shape of the closed form solutions of 
a $P$-solvable recurrence equation, for studying 
the degree and the dimension of invariant ideal.
We will provide degree estimates for the 
% algebraic variety of 
the invariant  ideal, which is useful for all invariant
generation methods which need a degree bound, like 
the proposed polynomial interpolation based method
and those in ~\cite{MS04a,MS04b,RK07b}.
Last but not least, we will investigate 
the dimension of the invariant ideal.
So that we can get a sufficient for non-trivial polynomial invariants
of a given $P$-solvable recurrence to exist.
Note that in our invariant generation method, we do not need (thus  never
compute) the closed form solutions explicitly.

% The very basic case (that is, the loop contains only assignments)
% to compute polynomial equational loop invariant is essentially compute all the 
% algebraic relation among the recurrence variables.
A ``solvable'' recurrence relation is, literally,  a recurrence relation which can
be solved by a closed formula depending only on the index number.
The $P$-solvable
recurrence relations have poly-geometrical expressions (Definition~\ref{def: poly-geo})
as closed form solutions, which is equivalent to 
the notion of {\em solvable mapping} 
in~\cite{RK04} or {\em solvable loop} in~\cite{Kov08} in
the respective contexts.

\ifCoRR{
\begin{Definition}[Univariate  $P$-solvable recurrence]
\label{defn: p-solvable-base}
Given a recurrence $R: \; x(n+1)=\lambda \, x(n) + f(n)$ 
in $\mathbb{K}$, if $f(n)$ is a poly-geometrical expression
in $n$ over $\mathbb{K}$, then $R$ is called \emph{univariate $P$-solvable recurrence}.
\end{Definition}

A multivariate recurrence is called  $P$-solvable recurrence,
if the recurrence variables can  essentially  (may need a linear coordinate change) be solved out one by one from $P$-solvable univariate recurrences 
We can define multivariate  $P$-solvable recurrence as follows.
}
\fi

\begin{Definition}[$P$-solvable recurrence]
\label{defn: p-solvable}
Let $n_1, \ldots, n_k$ be positive integers and define $s := n_1 + \cdots + n_k$.
% Let $\mathbb{F}$ be a field and 
Let $M$ be  a square matrix
over  $\mathbb{Q}$ of order $s$.
We assume that $M$ is block-diagonal with the following shape:
$$
M := 
\left(
\begin{array}{cccc}
\mathbf{M}_{n_1 \times n_1} & \mathbf{0}_{n_1 \times n_2} & \ddots & \mathbf{0}_{n_1\times n_k}\\
\mathbf{0}_{n_2 \times n_1}  & \mathbf{M}_{n_2 \times n_2} & \ddots & \mathbf{0}_{n_2\times n_k}\\
\ddots &  \ddots & \ddots & \ddots\\
\mathbf{0}_{n_k \times n_1}  & \mathbf{0}_{n_k \times n_2} & \ddots & \mathbf{M}_{n_k\times n_k}\\
\end{array}
\right)
.$$
Consider an $s$-variable recurrence relation $R$ 
in the variables $x_1, x_2, \ldots, x_s$ and with the following form:
$$
\begin{array}{c}
\left(
\begin{array}{c}
x_1(n+1)\\
x_2(n+1)\\
x_3(n+1)\\
\vdots\\
x_s(n+1)
\end{array}
\right)
 = 
M

\times 

\left(
\begin{array}{c}
x_1(n)\\
x_2(n)\\
x_3(n)\\
\vdots\\
x_s(n)\\
\end{array}
\right)

+

% \left(
% \begin{array}{c}
% % \mathbf{0}_{n_1 \times 1}\\
% {\mathbf{f}_1}_{n_1 \times 1}\\
% {\mathbf{f}_2}\left(x_1(n), \ldots,x_{n_1}(n)\right)_{n_2 \times 1}\\
% {\mathbf{f}_3}\left(x_1(n), \ldots,x_{n_1+n_2}(n)\right)_{ n_3 \times 1}\\
% \vdots\\
% {\mathbf{f}_k}\left(x_1(n), \ldots,x_{n_1+n_2+n_{k-1}}(n)\right)_{n_k \times 1}\\
% \end{array}
% \right)
% \end{array}

\left(
\begin{array}{c}
% \mathbf{0}_{n_1 \times 1}\\
{\mathbf{f}_1}_{n_1 \times 1}\\
{\mathbf{f}_2}_{n_2 \times 1}\\
{\mathbf{f}_3}_{ n_3 \times 1}\\
\vdots\\
{\mathbf{f}_k}_{n_k \times 1}\\
\end{array}
\right)
\end{array}
,$$
where 
$\mathbf{f}_1$ is a vector of length $n_1$ with coordinates in $\mathbb{Q}$ 
and where 
$\mathbf{f}_i$ is a tuple of length $n_i$ with coordinates
in the polynomial ring $\mathbb{Q}[x_1, \ldots, x_{n_1 + \cdots + n_{i-1}}]$,
for $i=2,\ldots,k$. 
Then, the recurrence relation $R$ is called \emph{$P$-solvable} over $\mathbb{Q}$ 
and the matrix $M$
is called the \emph{coefficient matrix} of $R$.
\end{Definition}

% Assume for each $i$, the maximal total degree for the polynomials in $\mathbf{f}_i$ 
% is $d_i$, then the $P$-solvable recurrence is said to in \emph{configuration} $(n_1,1), (n_2,d_2), \ldots, (n_k, d_k)$.

% establish a sufficient condition 
%(in Proposition~\ref{prop: sufficient-cond})
%as well as a necessary condition (in Proposition~\ref{prop: necessary-cond}) 
%for 

% A program or a loop can be modeled as a recurrence on all its variables,
% index by the executing steps of the program/loop.

% \begin{proof}
% \fbox{Complete!}
% Look into J\"urgen's PhD thesis book.
% \end{proof}

% Indeed, we can prove the two sets:
% poly-geometrical expressions and components of solvable recurrences are matchable.

% 
% \begin{Definition}[Univariate $P$-solvable recurrence: defn I] 
%     Let $h(n)$ be a poly-geometrical expression in $n$.
%     Let $\lambda$ be a complex number.
% % A series $x(n)$ defined by
% A recurrence relation
%  $x(n)=\lambda x(n-1) + h(n)$ is called a $P$-solvable recurrence.
% \end{Definition}
% 
% \begin{Definition}[Univariate $P$-solvable recurrence: defn II] 
% A recurrence relation $R$ is called a $P$-solvable if the 
% closed form solution
%  $x(n)$ is  poly-geometrical in $n$.
% \end{Definition}
% 

% \subsection{Solving a $P$-solvable recurrence in theory}

It is known that the solutions to $P$-solvable recurrences
are poly-geometrical expressions in $n$ w.r.t. the eigenvalues of 
the matrix $M$, see for example~\cite{RK04}. However, we need to  
% We aim at 
estimate the ``shape'', e.g. the degree % and height 
of those poly-geometrical expression
solutions, with the final goal of estimating the
``shape'' (e.g. degree, height, dimension)
% height and degree of 
the invariant ideal.
In this paper, we focus on degree and dimension estimates.
% which is sufficient 
% for the purpose of the 
% algorithm based on polynomial interpolation
% and presented in Section~\ref{sec:construction}.
% Estimating the height of the variety of the invariant ideal
% is work in progress that we shall report in a future article.

We first generalize 
the result of Lemma~\ref{lem: solvable-base} to the multi-variable case.

\ifCoRR
\begin{Proposition}
\label{prop: multi-Jordan-0}
Let $\alpha_1, \ldots, \alpha_m \in \overline{\mathbb{Q}}^* \setminus \{1\}$. 
Let $\lambda \in \overline{\mathbb{Q}}$ and
$M \in  \overline{\mathbb{Q}}^{s \times s}$ be a matrix in the following Jordan form 
$$ \left(
\begin{array}{c c c c c c}
\lambda		& 0 				& 0 		& \cdots 			& 0 & 0\\
1 	& \lambda 		& 0 		& \cdots			& 0 & 0\\
0				& 1 	& \lambda & 0 				& 0 & 0\\
0				& \ddots			& \ddots 	& \ddots 			&  \ddots & 0\\
0& 0 				& 0 				& \cdots 	& \lambda 	& 0  \\ 
0& 0 				& 0 				& \cdots 	& 1 	& \lambda \\
\end{array}
\right).$$
Consider an $s$-variable recurrence $R$ defined as follows:
$$X(n+1)_{s \times 1} = M_{s \times s}  \, X(n)_{s \times 1}  + F(n)_{s \times 1},\mbox{ where }$$ 
\begin{enumerate}[$(a)$]
	\item $X := x_1, x_2, \ldots, x_s$ are the recurrence variables;
	\item $F := (f_1, f_2, \ldots, f_s)$ is a list 
	of poly-geometrical expression in $n$ w.r.t. $\alpha_1, \ldots, \alpha_m$, 
	with maximal total degree $d$.
\end{enumerate}
Then we have:
\begin{enumerate}
\item if $\lambda=0$, then
$(f_1, f_1+f_2, \ldots, f_1+f_2+\cdots+f_{s})$
solves $R$. 
\item if $\lambda = 1$, then there exist $s$ 
poly-geometric expressions $(g_1,g_2, \ldots,g_s)$
in $\alpha_1, \ldots, \alpha_m$ such that 
for each $i \in 1 \cdots s$, 
$g_i$ is  a poly-geometrical expression
 		in $n$ w.r.t.
		$\alpha_1, \ldots, \alpha_m$
 		with total degree less or equal than $d+i$. 

\item if $\lambda \not \in \{0,1\}$, 
then there exists a solution
 of $R$, say $(y_1,y_2, \ldots,y_s)$,
 such that for each $i=1,\ldots,s$ we have
\begin{equation}
\label{eq:multi-Jordan-solution}
 y_i := c_i \lambda_i^n + g_i, \mbox{ where }
\end{equation}
for each $i \in 1\cdots s$:
  $(a)$ $c_i$ is a constant depending only on 
         the initial value of the recurrence; and
 	$(b)$ $g_i$  is like in the case of $\lambda=1$.
Moreover, assume further more that the following conditions hold:
\begin{enumerate}[$(i)$]
\item $\lambda$ is {\wm} independent w.r.t. $\alpha_1, \ldots, \alpha_m$;
\item $\deg(f_j,n)=0$ holds for all $j \in \{ 1,2,\ldots,s\}$.
\end{enumerate}
Then, for all $i=1,\ldots,s$, we can further choose $g_i$ such that 
$\deg(g_i, n) = 0$ holds and
the total degree of $g_i$ is less or equal than $\max(d,1)$. 
\end{enumerate}
\end{Proposition}
\begin{proof}
We observe that the recurrence variables of $R$ can be solved one after 
the other, from $x_1$ to $x_s$.
When $\lambda=0$, the conclusion is easy to verify. The case $\lambda \neq 0$
is easy to prove by
induction on $s$ with Lemma~\ref{lem:  solvable-base}.
\end{proof}
\fi

\begin{Proposition}
\label{prop: multi-Jordan}
Let $\lambda_1, \ldots, \lambda_s, \alpha_1, \ldots, \alpha_m \in \overline{\mathbb{Q}}^* \setminus \{1\}$. 
Let
$M \in  \overline{\mathbb{Q}}^{s \times s}$ be a matrix in the following Jordan form 
$$ \left(
\begin{array}{c c c c c c}
\lambda_1 		& 0 				& 0 		& \cdots 			& 0 & 0\\
\epsilon_{2,1} 	& \lambda_2 		& 0 		& \cdots			& 0 & 0\\
0				& \epsilon_{3,2} 	& \lambda_3 & 0 				& 0 & 0\\
0				& \ddots			& \ddots 	& \ddots 			&  \ddots & 0\\
0& 0 				& 0 				& \cdots 	& \lambda_{s-1} 	& 0  \\ 
0& 0 				& 0 				& \cdots 	&\epsilon_{s,s-1} 	& \lambda_s \\
\end{array}
\right),$$
 where for $i=2,\ldots, s$, $\epsilon_{i,i-1}$ is either $0$ or $1$.
Consider an $s$-variable recurrence $R$ defined as follows:
$$X(n+1)_{s \times 1} = M_{s \times s}  \, X(n)_{s \times 1}  + F(n)_{s \times 1},$$ 
where 
\begin{enumerate}
	\item $X := x_1, x_2, \ldots, x_s$ are the recurrence variables;
	\item $F := (f_1, f_2, \ldots, f_s)$ is a list 
	of poly-geometrical expression in $n$ w.r.t. $\alpha_1, \ldots, \alpha_m$, 
	with maximal total degree $d$.
\end{enumerate}
Then there exists a solution
 of $R$, say $(y_1,y_2, \ldots,y_s)$,
 such that for each $i=1,\ldots,s$ we have 
\begin{equation}
\label{eq:multi-Jordan-solution}
 y_i := c_i \lambda_i^n + g_i,
\end{equation}
where
  \begin{enumerate}
 	\item[$(a)$] $c_i$ is a constant depending only on 
         the initial value of the recurrence and
 	\item[$(b)$] $g_i$ is  a poly-geometrical expression
 		in $n$ w.r.t. \\
		$\lambda_1, \ldots, \lambda_{i-1}, \alpha_1, \ldots, \alpha_m$
 		with total degree less or equal than $d+i$. 
  \end{enumerate}
Assume further more that the following conditions hold:
\begin{enumerate}
\item[$(i)$] $\lambda_1, \lambda_2, \ldots, \lambda_s$ is {\wm} independent; %we have $\lambda_1 \neq 1$,
\item[$(ii)$] %$ \lambda_i$ is {\wm} independent w.r.t.\\
% $\{\lambda_1,\ldots,\lambda_{i-1}, 1\} \setminus \{ 0 \}$, for all $i=2,\ldots,s$ 
%      and,
% \item[$(iii)$] 
$\deg(f_j,n)=0$ holds for all $j \in \{ 1,2,\ldots,s\}$.
\end{enumerate}
Then, for all $i=1,\ldots,s$, we can further choose $y_i$ such that 
$\deg(g_i, n) = 0$ holds and
the total degree of $g_i$ is less or equal than $\max(d,1)$.
\end{Proposition}

\begin{proof}
We observe that the recurrence variables of $R$ can be solved one after 
the other, from $x_1$ to $x_s$. 
We proceed by induction on $s$.
The  case $s=1$ follows directly from Lemma~\ref{lem:  solvable-base}.
Assume from now on that $s>1$ holds and that we 
have found solutions $(y_1,y_2, \ldots, y_{s-1})$
for the first $s-1$ variables satisfying the requirements,
that is, Relation~(\ref{eq:multi-Jordan-solution})
with $(a)$ and $(b)$.
%% We first prove the first part of the conclusion.
We define
\begin{equation}
\tilde{f}(n) = f_s(n) - \epsilon_{s,s-1} \, y_{s-1}(n+1).
\end{equation}
Note that $\tilde{f}(n)$
is a poly-geometrical expression in $n$ w.r.t. 
$\lambda_1, \ldots, \lambda_{s-1}, \alpha_1, \ldots, \alpha_m$ with 
total degree less than or equal to $d+s-1$.
Moreover, 
for $v \in \{ n, \lambda_1^n, \ldots, \lambda_{s-1}^n,$ $\alpha_1^n, \ldots, \alpha_m^n \}$ 
we have 
\begin{equation}
\label{eq: 3.1}
\deg(\tilde{f}(n), v) \leq \max\left(\deg(f_s(n),v), \deg(y_{s-1}(n),v)\right).
\end{equation}
It remains to solve $x_s$ from
	\begin{equation}
		\label{eq: recur-x_i}
		x_{s}(n+1) = \lambda_s \, x_s(n) + \tilde{f}(n)
		\end{equation}
	in order to solve all the variables $x_1, \ldots, x_s$.
	Again, by Lemma~\ref{lem: solvable-base}, there exists a poly-geometrical expression $$y_s := c_s\,\lambda_s^n + g_s(n),$$ 
	where $g_s(n)$ is poly-geometrical expression
	in $n$ w.r.t. \\
	$\lambda_1, \ldots, \lambda_{s-1}, \alpha_1, \ldots, \alpha_m$, of total degree upper bounded by $d+s$.
This completes the proof of the properties $(a)$ and $(b)$ for $y_s$.

Now we assume that $(i)$, $(ii)$ hold % and $(iii)$ hold
and we prove the second half of the conclusion.
Observe that we have $\deg(g_s(n),n) = \deg(\tilde{f}(n), n)$, which is $0$,
according to Relation~(\ref{eq: 3.1}) and the fact that
we can choose $y_{s-1}$ such that  $\deg(y_{s-1}(n),n)=0$ holds.
Next, we observe that 
for each $$v \in \{n, \lambda_1^n, \ldots, \lambda_{s-1}^n, \alpha_1^n, \ldots, \alpha_m^n\},$$
we have $\deg(g_s(n),v)=\deg(\tilde{f}(n),v)$, which is less or equal to 
$\deg(y_{s-1}(n),v)$ by Relation~(\ref{eq: 3.1}). 
Therefore, the total degree of $g_s$ is less or equal 
than the total degree of $y_{s-1}$,
which is less or equal than $\max(d,1)$ by our induction hypothesis.
This completes the proof.
\end{proof}

\begin{Theorem}
\label{thm: shape-degree}
Let $R$ be a $P$-solvable recurrence relation. Using the 
same notations $M, k, s, F, n_1, n_2, \ldots, n_k$ as 
in  Definition~\ref{defn: p-solvable}.
Assume $M$ is in a Jordan form.
Assume
the eigenvalues $\lambda_1, \ldots, \lambda_{s}$ 
of $M$ (counted with multiplicities)  are different from $0,1$,
with $\lambda_i$ being the $i$-th diagonal element of $M$.
Assume for each block $j$
%\ifCoRR{
%\begin{itemize}
%\item there are   $r_j$ non-zero eigenvalues, which appear before $0$s;
%\item
%\end{itemize}  
%}
%\else
%
%the total degree of any polynomial in 
%$\mathbf{f}_j$ (for $i=2 \cdots k$) is upper bounded by $d_j$.
%% \end{itemize}
%\fi
the total degree of any polynomial in 
$\mathbf{f}_j$ (for $i=2 \cdots k$) is upper bounded by $d_j$.
For each $i$, we denote by $b(i)$ the block number of the index $i$,
that is,  
\begin{equation}
\sum_{j=1}^{b(i)-1} \, n_j < i \leq \sum_{j=1}^{b(i)} \, n_j.
\end{equation}
Let $D_1 := n_1$ and for $all j \in \{ 2, \ldots,  k\}$
let $D_{j} : = d_{j} \, D_{j-1}+n_j.$ 
Then, there exists a solution
$(y_1,y_2, \ldots,y_s)$ for $R$
of the following form:
\begin{equation}
y_i := c_i \lambda_i^n + g_i,
\end{equation}
for all $i \in 1 \cdots s$,  where
  \begin{enumerate}
 	\item[$(a)$] $c_i$ is a constant depending only on the initial value of the recurrence;
 	\item[$(b)$] $g_i$ is  a poly-geometrical expression
 		in $n$ w.r.t. $\lambda_1, \ldots, \lambda_{i-1}$,
		and with total degree less or equal than $D_{b(i)}$. 
  \end{enumerate}
Moreover, 
if $\{\lambda_1,\ldots,\lambda_{s}\}$
is {\wm} independent,
% if $\lambda_1 \neq 1$ holds 
% and if for all $i=2,\ldots,s$ the eigenvalue 
% $ \lambda_i$ is {\wm} independent w.r.t. 
% $\{\lambda_1,\ldots,\lambda_{i-1}, 1\} \setminus \{ 0 \}$,
then, for all $i=1,\ldots,k$, we can further choose $y_i$ 
such that $\deg(g_i, n) = 0$ holds and
the total degree of $g_i$ is less or equal 
than $\prod_{2 \leq  t \leq b(i)} \max(d_t,1)$.
\end{Theorem}
\begin{proof}
We proceed by induction on the number of blocks, that is, $k$.
The case $k=1$ follows immediately from Proposition~\ref{prop: multi-Jordan}.
Assume from now on that the conclusion holds for a value 
$k=\ell$, with ${\ell} \geq 1$
and let us prove that it also holds for $k = \ell +1$.
We apply the induction hypothesis to solve the first $\ell$ blocks of variables,
and suppose that $\mathbf{y}_{\ell}$ is a solution satisfying 
the properties in the conclusion.
For solving the variables in the $(\ell + 1)$-th block, we substitute 
$\mathbf{y}_{\ell}$ to $f_{\ell + 1}$ and obtain a tuple 
of poly-geometrical expressions in $n$
w.r.t the eigenvalues of the first $\ell$ blocks and
 with total degree bounded by  $d_{\ell} \, D_{\ell}$.
Therefore, applying again Proposition~\ref{prop: multi-Jordan}, we can find solutions 
for the variables in the  $(\ell + 1)$-th block 
satisfying the properties required in the 
conclusion. This completes the proof.
\end{proof}

Note that the degree estimate in Theorem~\ref{thm: shape-degree} 
depends on how the block structure of the recurrence is exploited,
for example, a $2 \times 2$ diagonal matrix can be viewed as 
a matrix with a single block or a matrix with two $1 \times 1$ diagonal blocks.
\ifCoRR
\else
Please refer to~\cite{DBLP:journals/corr/abs-1201-5086},
%an extended version of this article,  
for finer results on
degree estimates of the poly-geometric expression solutions, 
 with $0$ and $1$ being eigenvalues.  
\fi

\ifCoRR
In practice, one might want to decouple the recurrence first, and then
study the recurrence variable one by one (after a linear coordinate change) 
to get better degree estimates for the
poly-geometrical expression solutions,
regarded as polynomials of $n$-exponential terms as the eigenvalues of 
the coefficient matrix.
% In general, the finer the structure is exploit (ideally,
% we should study the recurrence variable one by one), the better bound would be obtained.
We will just use a simple example to illustrate this idea.
% decouple and max block
\begin{Example}
Consider the recurrence:
$$
\left(
\begin{array}{c}
x(n+1)\\
y(n+1)\\
z(n+1)
\end{array}
\right)
  \; := \; 
  \left(
\begin{array}{ccc}
2 & 0 & 0\\
0 & 3 & 0\\
0 & 0 & 3
\end{array}
\right)
\,
 \times
\,
    \left(
\begin{array}{c}
x(n)\\
y(n)\\
z(n)
\end{array}
\right)
\;
+ 
\;
  \left(
\begin{array}{c}
0\\
x(n)^2\\
x(n)^3
\end{array}
\right)
$$

Viewing the recurrence as two blocks $(x)$ and $(y,z)$,
the degree estimate according to Theorem~\ref{thm: shape-degree} 
would be bounded by $5$  ($3 \times 1 + 2$).

If we decouple the $(y,z)$ block to the following two recurrences
$$ y(n+1) = 3\, y(n) + x(n)^2 \mbox{ and } z(n+1) = 3\, z(n) + x(n)^3 ,$$

the we can easily deduce that the degree of the poly-geometrical expression 
for $y$ and $z$ are upper bounded by $2$ and $3$ respectively, again according 
to  Theorem~\ref{thm: shape-degree}.
%  $(y,z)$ block to two 
%
\end{Example}
\fi

It is easy to generalize the previous results to the case
of a matrix $M$ which is not in Jordan form.
Let $Q$ be a non-singular matrix such that  
$J :=  Q \, M \, Q^{-1}$ is a Jordan form of $M$.
Let the original recurrence $R$ be
$$X(n+1) = M \, X(n) + F.$$
Consider the following recurrence $R_{Q}$  
           $$Y(n+1) = J \,  Y(n) + Q F.$$
It is easy to check that
if $$\left(y_1(n), y_2(n), \ldots, y_s(n)\right)$$
solves  $R_{Q}$, 
then $$Q^{-1} \, \left(y_1(n), y_2(n), \ldots, y_s(n)\right)$$
solves $R$.
Note that 
an invertible matrix over $\overline{\mathbb{Q}}$  maps a
tuple of poly-geometrical expressions to another tuple of 
poly-geometrical expressions; moreover it preserves 
the highest degree among the expressions in the tuple.

We turn now our attention to the question of estimating 
the degree of the invariant ideal of a $P$-solvable recurrence relation.

\ifCoRR
\begin{Proposition}
\label{prop: iv-R}
Let $R$ be an $s$-variable $P$-solvable recurrence relation,
with recurrence variables $(x_1,x_2, \ldots,x_s)$. 
Let $\mathcal{I} \subset {\mathbb{Q}}[x_1,x_2, \ldots, x_s]$ 
be the invariant ideal of $R$.
Denote by $\mathcal{I}^e$ the extension of $\mathcal{I}$ in 
$\overline{\mathbb{Q}}[x_1,x_2, \ldots, x_s]$.
Let $A = \alpha_1, \alpha_2, \ldots, \alpha_s$ be
% be a sequence of $k$ elements in $\overline{\mathbb{Q}}$.
the eigenvalues (counted with multiplicities) of the coefficient matrix of $R$.
Let  $\mathcal{M}$ be the multiplicative relation ideal of
 $A$  associated with variables $y_1,\ldots, y_s$. 
Then, there exists a sequence of $s$ poly-geometrical expressions 
in $n$ w.r.t. 
$\alpha_1, \alpha_2, \ldots, \alpha_s$, say 
$$f_1(n, \alpha_1^n, \ldots, \alpha_k^n), \ldots, f_s(n, \alpha_1^n, \ldots, \alpha_k^n),$$ which solves $R$.
Moreover, 
 we have
%$$ \mathcal{I}^e \, = \, \left( \langle x_1 - f_1(n, y_1, \ldots, y_k), \ldots, x_s-f_s(n, y_1, \ldots, y_k))\rangle + \mathcal{M} \right) \;  \cup \; \overline{\mathbb{Q}}[x_1,x_2, \ldots, x_s].$$
$$ \mathcal{I}^e \, = \, \left( \mathcal{S} +\mathcal{M} \right) 
\;  \cap \; \overline{\mathbb{Q}}[x_1,x_2, \ldots, x_s],$$
where $\mathcal{S}$ is the ideal generated by
$\langle x_1 - f_1(n, y_1, \ldots, y_s), \ldots, x_s-f_s(n, y_1, \ldots, y_s)$ in $\overline{\mathbb{Q}}[x_1,x_2, \ldots, x_s,n,y_1,\ldots,y_s]$.
\end{Proposition}
\begin{proof}
The existence of $f_1, f_2, \ldots, f_s$ follows by Theorem~\ref{thm: shape-degree} and the fact that linear combination of poly-geometrical expressions 
w.r.t. $n$ are still poly-geometrical expressions.
The conclusion follows from
Lemma~\ref{lem: transcent}.
% implies the following relation
% \begin{equation}
%  \mathcal{I}^e \, = \, \left( \mathcal{S} +\mathcal{M} \right)
% \;  \cap \; \overline{\mathbb{Q}}[x_1,x_2, \ldots, x_s].
% \end{equation}
\end{proof}
\fi

\ifCoRR{
The following lemma is not hard to prove and one can find 
a proof in~\cite{KZ08}.
\begin{Lemma}
\label{lem: iv-ext}
Let $R$ be a $P$-solvable recurrence relation
defining $s$ sequences in $\mathbb{Q}^s$,
with recurrence variables $(x_1,x_2, \ldots,x_s)$. 
Let $\mathcal{I}$ be the invariant ideal of $R$
in  ${\mathbb{Q}}[x_1,x_2, \ldots, x_s]$;
and let $\overline{\mathcal{I}}$  be  the invariant ideal of $R$
in  $\overline{\mathbb{Q}}[x_1,x_2, \ldots, x_s]$.
Then  $\overline{\mathcal{I}}$ equals to $\mathcal{I}^e$, the extension of $\mathcal{I}$ in 
$\overline{\mathbb{Q}}[x_1,x_2, \ldots, x_s]$.
\end{Lemma}

With Proposition~\ref{prop: iv-R}
and Proposition~\ref{prop: bi-rational-degree}, we are able 
to estimate the degree of polynomials in a generating system 
of the invariant ideals. 
Now we are able to estimate the total degree of closed form 
solutions of a $P$-solvable recurrence without solving the 
recurrence explicitly.
\fi

%  of a large category of $P$-solvable recurrence relations,
% say in the degree of $n$ in the solutions (viewed as poly-geometrical expressions)
% is $0$.
\begin{Theorem}
\label{thm: main}
Let $R$ be a $P$-solvable recurrence relation
defining $s$ sequences in $\mathbb{Q}^s$,
with recurrence variables $(x_1,x_2, \ldots,x_s)$. 
% Let $R$ be an $s$-variable $P$-solvable recurrence relation,
% with  variables $(x_1,x_2, \ldots,x_s)$. 
Let $\mathcal{I} \subset {\mathbb{Q}}[x_1,x_2, \ldots, x_s]$ 
be the invariant ideal of $R$. 
Let $A = \alpha_1, \alpha_2, \ldots, \alpha_s$ be 
% be a sequence of $k$ elements in $\overline{\mathbb{Q}}$.
the eigenvalues (counted with multiplicities) of the coefficient matrix of $R$.
Let  $\mathcal{M}$ be the multiplicative relation ideal of
 $A$  associated with variables $y_1,\ldots, y_k$. 
 Let $r$ be the dimension of  $\mathcal{M}$.
Let $f_1(n, \alpha_1^n, \ldots, \alpha_k^n), \ldots, f_s(n, \alpha_1^n, 
\ldots, \alpha_k^n)$ be a sequence of $s$ poly-geometrical expressions 
in $n$ w.r.t. $\alpha_1, \alpha_2, \ldots, \alpha_s$  that solves $R$.
% Let $\mathcal{I}$ be the invariant ideal of $R$.
Suppose $R$ has a $k$ block configuration as 
$(n_1,1), (n_2, d_2), \ldots, (n_k, d_k)$. 
Let $D_1 := n_1$; and for all $j \in \{ 2, \ldots,  k\}$,
let $D_{j} : = d_{j} \, D_{j-1}+n_j$.  
Then % the invariant ideal of $R$ has degree upper bounded by
we have 
$$\deg(\mathcal{I}) \leq \deg(\mathcal{M})\, D_k^{r+1}.$$ 
Moreover, if 
the degrees of $n$ in $f_i$ ($i=1\cdots s$) are $0$, then we have $$\deg(\mathcal{I}) \leq \deg(\mathcal{M})\, D_k^{r}.$$

%%%and satisfies for $\forall \, i \in 1.. s$, $\deg(f_i, n) = 0$ holds.
%%%% 
%%%Then we have 
%%%$$ \deg(\mathcal{I}) \leq  \deg(\mathcal{M})\, (D_k)^r.$$
%%% Denote by $D$ the degree of $\mathcal{M}$.
%%%Moreover, there exists a finite set of polynomials in $\mathbb{Q}[x_1,x_2, \ldots,x_s]$
%%%generating $\mathcal{I}$ and such that each of these polynomials
%%%has total degree less than or equal to $\deg(\mathcal{M})$.
\end{Theorem}
\begin{proof}

Denoting by ${\Pi}$ the standard projection from
$ { \overline{\mathbb{Q}} }^{s+1+s} $ to $ { \overline{\mathbb{Q}} }^{s}$:
$$(x_1,x_2, \ldots, x_s,n,y_1,\ldots,y_s) \mapsto (x_1,x_2, \ldots, x_s),$$
we deduce 
\ifCoRR
 by Proposition~\ref{prop: iv-R} 
\fi
that
\begin{equation}
V({\mathcal I}) = \overline{ {\Pi}( V(\mathcal{S} +\mathcal{M} ))},
\end{equation}
where $\mathcal{S}$ is the ideal generated by
$\langle x_1 - f_1(n, y_1, \ldots, y_s), \ldots, x_s-f_s(n, y_1, \ldots, y_s)$ in $\overline{\mathbb{Q}}[x_1,x_2, \ldots, x_s,n,y_1,\ldots,y_s]$.

Thus, 
\ifCoRR
 by Lemma~\ref{lem: degree-lm}, 
\fi
we have 
$$\deg(\mathcal{I}) \leq \deg(\mathcal{S} + \mathcal{M}).$$ 

It follows from
\ifCoRR
 Proposition~\ref{prop: bi-rational-degree}
\else
the special shape of $S$ and 
Bezout bound inequality (see details in~\cite{DBLP:journals/corr/abs-1201-5086})
\fi
that
 $$\deg(\mathcal{S} + \mathcal{M}) \leq \deg(\mathcal{M})\, D_k^{r+1},$$
since the total degree of $f_i$ of $R$ is bounded by $D_k$ according to
Theorem~\ref{thm: shape-degree} and the dimension of $\mathcal{M}$ is $r+1$
is in $\mathbb{Q}[n,y_1,\ldots,y_s]$.
%%%and from the shape of the generators of $\mathcal{S}$, that we have
%%%$$\deg(\mathcal{S} + \mathcal{M}) = \deg(\mathcal{M}).$$
%%%Therefore, we have $\deg(\mathcal{I}) \leq  \deg(\mathcal{M})$.
%%%This completes the proof.
%%
%%
% Denote by $D$ the degree of $\mathcal{M}$.
% Denote by $\mathcal{I}^e$ the extension of $\mathcal{I}$ in
% $\tilde{\mathbb{Q}}[x_1,x_2,\ldots,x_s]$.
% It follows from Theorem~\ref{prop: iv-R} that 
% there exists a finite set of polynomials in $\tilde{\mathbb{Q}}[x_1,x_2,\ldots,x_s]$ 
% generating  $\mathcal{I}^e$ and each of these polynomials
% has total degree less than or equal to $D$.
% Therefore, by virtue of Lemma~\ref{lem: Galois-contract}, 
% the ideal $\mathcal{I}$ can be generated by   polynomials of
% $\mathbb{Q}[x_1,x_2, \ldots,x_s]$, each of which 
% with total degree less than or equal to $\delta \, D$.

With similar arguments,
the second part of the conclusion follows from the fact that $S+M$ can be viewed as an ideal in
in $\overline{\mathbb{Q}}[x_1,x_2, \ldots, x_s,n,y_1,\ldots,y_s]$,
where $M$ has dimension $r$.
%\fbox{Another point}.
\end{proof}

Indeed, the degree bound in Theorem~\ref{thm: main} is ``sharp'' in the sense that 
it is reached by many of the examples we have considered. 
\ifCoRR
Let show two of such examples below.
% can always be reached, is sharp, see Example~\ref{ex: Fibonacci}. 
% One more surprising and interesting consequence is that, the degree 
% bound does not dependent on the degree of non-linear part of the related recurrence.

\begin{Example}[Example~\ref{ex: intro2} Cont.]
\label{ex: Fibonacci}
% Using total degree $4$ and $20$ points, both Algorithm~\ref{alg: pii} and % Algorithm~\ref{alg: mii} 
% compute within $0.15s$ the ideal $$\langle 1-y^4+2 x y^3+x^2 y^2-2 x^3 y-x^4 \rangle,$$
% which is indeed the invariant ideal of $(x,y)$,
% since the invariant ideal is positive dimensional thus of dimension $1$.
The corresponding recurrence  only $1$ block. 
Denote by $A := \frac{-\sqrt{5}+1}{2}, \frac{\sqrt{5}+1}{2}$.
One can easily check that $A$ is $\wm$ independent.
Note the multiplicative relation ideal of $A$ % $\frac{-\sqrt{5}+1}{2}, \frac{\sqrt{5}+1}{2}$
associated with variables $u,v$ is generated by $u^2 v^2-1$ and thus has degree $4$ and dimension $1$ in $\mathbb{Q}[u,v]$. 
Therefore, by Theorem~\ref{thm: main}, the degree of invariant ideal
bounded by $4 \times 1^1$. 
This implies that the degree bound given  
by Theorem~\ref{thm: main} is sharp.
\end{Example}
\else
Please refer to~\cite{DBLP:journals/corr/abs-1201-5086}
for such examples.
\fi

% \subsection{Sufficient or necessary conditions for non-trivial polynomial invariants}
In the rest of this section,
we are going to investigate the dimension of the invariant 
ideal of $P$-solvable recurrences.
This can help to answer the following 
natural question: whether or not 
the invariant ideal 
of a $P$-solvable recurrence over $\mathbb{Q}$
% a field $\mathbb{F}$  
is the trivial ideal of  $\mathbb{Q}[x_1,\ldots,x_s]$?
% In general, an ideal I in $\overline{\mathbb{Q}}[x_1,\ldots,x_s]$
% also contains elements in $\mathbb{Q}[x_1,\ldots,x_s]$!
% f+g*i \in I -> f^2+g^2 \in I \cap \mathbb{Q}[x_1,\ldots,x_s]
Note that it is obvious that the invariant ideal is not 
the whole polynomial ring. 
% , since there are points
% in the least

\begin{Theorem}
    \label{thm: dimension}
Using the same notations 
as in  Definition~\ref{defn: p-solvable}.
Let $\lambda_1,\lambda_2,\ldots, \lambda_s$
be the eigenvalues of  $M$ 
counted with multiplicities.
% Let $r$ be the dimension of 
Let $\mathcal{M}$ be the multiplicative relation ideal 
of $\lambda_1,\lambda_2, \ldots, \lambda_s$. 
Let $r$ be the dimension of $\mathcal{M}$. 
Let $\mathcal{I}$ be 
the invariant ideal of $R$.
Then $\mathcal{I}$ is of dimension at most $r+1$.
Moreover, for generic initial values, 
\begin{enumerate}
	\item the dimension of $\mathcal{I}$ is at least $r$;
	\item if $0$ is not an eigenvalue of $M$ and
	$\lambda_1,\lambda_2, \ldots, \lambda_s$ is {\wm} independent, then
	$\mathcal{I}$
has dimension $r$.
% can we say more precisely when the {\wm}-independent condition fails?
\end{enumerate}
\end{Theorem}
\begin{proof}
Assume without loss of genericity that $M$ is in Jordan form.
By Theorem~\ref{thm: shape-degree},
we deduce that $R$ has a solution $(f_1, f_2, \ldots, f_s)$ 
as follows
$$\left(c_1 \, \lambda_1^n+ h_1(n), c_2 \, \lambda_2^n+ h_2(n), \ldots, 
c_s \, \lambda_s^n+ h_s(n)\right),$$
where for each $i \in 1 \cdots s$, $c_i$ is a constant in $\overline{\mathbb{Q}}$ depending only on the initial value
of $R$, and $h_i$ is a poly-geometrical expression in $n$
w.r.t. $\lambda_1, \ldots, \lambda_{i-1}$. 
Moreover, we have
\begin{enumerate}
	\item for generic initial values, none 
          of $c_1, c_2, \ldots, c_s$ is $0$; 
	\item if the eigenvalues of $M$ can be ordered in
	$\lambda_1,\lambda_2, \ldots, \lambda_s$ s.t.
	$\lambda_1 \neq 1$ and for 
	each $i \in 2 \cdots s$,  
	$\lambda_i$ is 
	{\wm} independent w.r.t. $\lambda_1,\lambda_2, \ldots, \lambda_{i-1}$, 
	then we can require that, for all $i \in 1 \cdots s$, we have $\deg(f_i,n)=0$.
\end{enumerate}

Viewing $n$, $\lambda_i^n$ (for $i=1,\ldots,s$) as indeterminates, 
let us associate coordinate variable $u_0$ to $n$, $u_i$ to  $\lambda_i^n$  
(for $i=1,\ldots,s$).
Denote by $V$ the variety of 
$\mathcal{I}$ in $\overline{\mathbb{Q}}^{s}$ 
(with coordinates $x_1, x_2, \ldots, x_s$).
% and 
% in  $\overline{\mathbb{Q}}^{s+1}$  (with coordinates $u_0, x_1, x_2, \ldots, x_s$).
Note that we have 
$$\dim(V) = \dim(\mathcal{I}).$$
% \mbox{ and } \dim(V_2) \leq \dim(\mathcal{I})+1.$$

Denote by $W_1, W_2$ respectively the variety of 
$\mathcal{M}$ in $\overline{\mathbb{Q}}^{s}$ 
(with coordinates $u_1, u_2, \ldots, u_s$) and 
in  $\overline{\mathbb{Q}}^{s+1}$  (with coordinates $u_0, u_1, u_2, \ldots, u_s$).
Note that we have 
$$\dim(W_1) = r \mbox{ and } \dim(W_2) = r+1.$$

Consider first
    the map $F_0$ defined below:
    $$
    \begin{array}{c}
        F_0: \overline{\mathbb{Q}}^{s+1} \mapsto \overline{\mathbb{Q}}^{s+1} \\
           % (u_0, u_1, \ldots, u_s) \to (c_1\, u_1 + f_1(u_0), \ldots, c_s\, u_s + f_s(u_0, u_1, \ldots, u_{s-1})).
           (u_0, u_1, \ldots, u_s) \to (c_1\, u_1 + f_1, \ldots, c_s\, u_s + f_s).
    \end{array}
    $$
By Theorem~\ref{thm: main}, we have
$V = \overline{F_0(W_2)}$.
Therefore, we have we have $\dim(\mathcal{I}) = \dim(V) \leq  \dim(W_2) = r+1$.

Now assume the initial value
of $R$ is generic, thus we have $c_i \neq 0$, for all $i \in 1 \cdots s$.
Let us consider
    the map $F_1$ defined below:
    $$
    \begin{array}{c}
        F_1: \overline{\mathbb{Q}}^{s+1} \mapsto \overline{\mathbb{Q}}^{s+1} \\
           (u_0, u_1, \ldots, u_s) \to (u_0, c_1\, u_1 + f_1, \ldots, c_s\, u_s + f_s).
    \end{array}
    $$
Let us denote by $V_2$ the variety $\overline{F_1(W_2)}$.
By virtue of Theorem~\ref{thm: main}, we  have $\dim(V_2) = \dim(W_2) = r+1$.
Denote by $\Pi$ the standard projection map that
forgets the first coordinate, that is, $u_0$.
We observe that $V = \overline{\Pi(V_2)}$.
Therefore, we have $\dim(V) \geq \dim(\overline{\Pi(V_2)}) -1 =r$.

Now we further assume 
	$\lambda_1 \neq 1$ and for 
	each $i \in 2 \cdots s$,  
	$\lambda_i$ is 
	{\wm} independent w.r.t. $\lambda_1,\lambda_2, \ldots, \lambda_{i-1}$
	the invariant ideal of $R$. 
	In this case, we have that for all $i \in 1 \cdots s$,  $\deg(f_i,n)=0$.
Let us consider
    the map $F_2$ defined below:
    $$
    \begin{array}{c}
        F_2: \overline{\mathbb{Q}}^{s} \mapsto \overline{\mathbb{Q}}^{s} \\
           (u_1, \ldots, u_s) \to (c_1\, u_1 + f_1, c_2\, u_2 + f_2, \ldots, c_s\, u_s + f_s).
    \end{array}
    $$
By Theorem~\ref{thm: main}, we have 
$V = \overline{F_2(W_1)}$. % and thus $W_2 = \overline{F_2^{-1}(V)}$ since $F_2$ is invertible.
% 
% $F_2$ defines an invertible polynomial map
% between the variety of $\mathcal{M}$ in $\overline{\mathbb{Q}}^{s}$, which is of dimension
% $r$, 
% to the variety of $\mathcal{I}$ in $\overline{\mathbb{Q}}^{s}$.
Therefore, we have $\dim(\mathcal{I}) = \dim(V) = \dim(W_1) = r$. This completes the proof.
\end{proof}

The following result, which is a direct consequence of 
Theorem~\ref{thm: dimension}, can serve as a sufficient
condition for the invariant ideal to be non-trivial.
This condition is often satisfied when there are 
eigenvalues with multiplicities or when $0$ and $1$ 
are among the eigenvalues.
\begin{Corollary}
\label{coro:criterion}
Using the same notations as in Theorem~\ref{thm: dimension}.
If $r+1<s$ holds, then $\mathcal{I}$ is not the zero ideal
in $\mathbb{Q}[x_1, x_2,\ldots, x_s]$.
\end{Corollary}

The following corollary indicates that, % we can determine that 
the fact that the inductive loop invariant is trivia
could be determined by just investigating the multiplicative
relation among the eigenvalues of the underlying recurrence. 
\begin{Corollary}
\label{cor: not-exist}
Using the same notations as in Theorem~\ref{thm: dimension},
consider the corresponding loop $\mathcal{L}$ with $x_1(0):=a_1, \ldots, x_s(0):=a_s$,
where $a_1,\ldots, a_s$ are indeterminates. 
If the eigenvalues of $R$ are multiplicatively independent,
% and all the initial values are indeterminate. 
then the inductive invariant ideal of $\mathcal{L}$ is
% trivial $\mathcal{I}$ is not 
the zero ideal
in $\mathbb{Q}[a_1,\ldots,a_s, x_1, x_2,\ldots, x_s]$.
\end{Corollary}
\begin{proof}
% Note that$R$ is a recurrence relation defined 
% on $\mathbb{Q}(a_1,\ldots,a_s)[x_1,\ldots,x_s]$.
Since
there is only trivial multiplicative relation,
 the multiplicative relation ideal of the eigenvalues is $0$, which is of dimension $s$. 
By Theorem~\ref{thm: dimension}, the invariant of $R$ must be zero ideal in 
$\mathbb{Q}(a_1,\ldots,a_s)[x_1,\ldots,x_s]$, since its dimension
must be at least $s$.

Assume there exists a non-zero invariant polynomial $p$ of $\mathcal{L}$,
then $p$ must be an invariant polynomial of $R$ 
since the loop variables $a_1,\ldots,a_s$ are free to take
any value. This is a contradiction to the fact that the invariant ideal of $R$ is 
trivial. Therefore,  the inductive invariant ideal of $\mathcal{L}$ is
% trivial $\mathcal{I}$ is not 
the zero ideal
in $\mathbb{Q}[a_1,\ldots,a_s, x_1, x_2,\ldots, x_s]$.
% $\mathbb{Q}(a_1,\ldots,a_s)[x_1,\ldots,x_s]$contain both variables from 
\end{proof}

\ifCoRR
\begin{Example}
Consider the recurrence:
$$(x(n+1),y(n+1)) := (3 \, x(n) +y(n), 2\,y(n)) \mbox{ with } x(0)=a, y(0)=b.$$
On one hand, the two eigenvalues  are $2$ and $3$ which are multiplicatively
independent, therefore, by Corollary~\ref{cor: not-exist},
the invariant ideal of the corresponding loop is trivial.

On the other hand, % let $i$ be a loop counter, 
for loop variables  $(a,b,x,y)$, 
the reachable set of the loop is
$$\mathfrak{R}:=\{(a,\, b,\, (a+b)\,3^i - b\,2^i, \, b\,2^i) \mid  (a,b) \in \mathbb{Q}^2, \; i \mbox{ is a non-negative integer} \}.$$
Therefore,  according to Lemma~\ref{lem: transcent}, any polynomial vanishes on all points  of
$\mathfrak{R}$ must be $0$.
\end{Example}
\fi
Note in Theorem~\ref{thm: dimension}, if we 
drop the ``generic'' assumption on
the initial values, then the conclusion
might not hold. 
\ifCoRR
\else
Please refer to~\cite{DBLP:journals/corr/abs-1201-5086}
for an example illustrating when all the eigenvalues
are different and multiplicatively independent but 
the invariant ideal is not trivial.
\fi
\ifCoRR
The following example illustrate 
this for the case when all the eigenvalues
are different and multiplicatively independent, but 
the invariant ideal is not trivial.
\begin{Example}
Consider the linear recurrence 
$x(n+1) = 3 \, x(n) - y(n), y(n+1)=2 \, y(n)$
with $(x(0),y(0))=(a,b)$.
The eigenvalues of the coefficient
matrix are $2,3$, which are multiplicatively independent.
One can check that,
when $a=b$, the invariant ideal is generated 
by $x-y$. However, generically,
that is when $a \neq b$ holds, 
the invariant ideal is the zero ideal.
% y(n)=2^n
% rsolve({x(n+1)=3*x(n)-2^n, x(0)=a}, x);
% find a =1 is a bad initial value. 
\end{Example}
\fi
% # M := Matrix([[0,0,1],[1,0,-3],[0,1,3]]);
% # f := CharacteristicPolynomial(M,x);
% # f := (x-1)^3;
% # M := CompanionMatrix(f, x);
% # MI has dimension r=0, but the RI has dimension r+1, that is 1.
% # st := time(); TRDrecurrence_invariants_interpolate([x,y,z],[z,x-3*z, y+3*z], [1,2,3]); time()-st;
% #                              2              2
% #             [x + y + z - 6, y  + 4 y z + 4 z  - 6 y - 24 z + 20]
% #

% \input{construction}

% !Tex root=mx-2012-casc.tex
\section{Algorithm and experimental results}
% \section{Computing invariants via interpolation}
\label{sec:construction}
% variety of the invariant ideal
In this section, we shall discuss how to 
compute invariant ideals of $P$-solvable recurrences as 
well as polynomial loop invariants.
Our approach is based on polynomial interpolation
and consists essentially of three main steps.
\begin{enumerate}
\item Sample a list of points $S$
      from the trajectory of the recurrence or loop.
\item Compute all the polynomials vanishing on $S$ 
      up to a certain degree,
      which can be either a known degree bound or a ``guessed'' bound.
\item Check whether or not the interpolated
      polynomials are invariants of the loop. % indeed generate the invariant ideal.
\end{enumerate}
% Therefore, one key problem is how to certify
% the polynomial equations obtained by interpolation
% are indeed invariants.

As one can see from our algorithm sketch, we need to
check whether or not a given condition (say a polynomial equation
or a polynomial inequality) is an invariant.
In general, roughly speaking, when a branch condition 
contains constraints given by inequalities, the problem of checking whether or not 
a linear equation is a loop invariant is undecidable, see~\cite{MS04b}
for a more detailed discussion.  Nevertheless, criteria showing that
a given condition is indeed an invariant are useful in practice.
For this reason, we are interesting necessary or sufficient conditions
for a conjunction of polynomial equations to be an invariant of a loop.

\subsection{Checking invariants}

Proposition~\ref{prop: iv-prop} states a necessary
condition for a set of polynomials 
to be the invariant ideal of a given loop.

\begin{Proposition}
\label{prop: iv-prop}
Given a loop $\mathcal{L}$ with only one branch
and let $A$ be the assignment function.
Let $I$ be the inductive invariant ideal of $\mathcal{L}$.
Then for any point $\alpha \in V(I)$, we have 
$A(\alpha) \in V(I)$.
\end{Proposition}

\ifCoRR
\begin{proof}	
	Denote by $\mathcal{T}$ the
 inductive trajectory of $\mathcal{L}$.
	Let $W$ be the Zariski closure of $A(V(I))$.
	% Let $W_1$ be the Zariski closure of $A(\mathfrak{T}(\mathcal{L}))$.
	Let $W_{1}$ be the Zariski closure of $W \setminus V(I)$.
	We proceed by contradiction, thus we assume
	$W_{1} \neq \emptyset$.
	Then we have $V(I) = A^{-1}(W_{1}) \cup A^{-1}(V(I) \cap W)$ and $\mathcal{T} \subseteq A^{-1}(V(I) \cap W)\;
	\mbox{ and } \; \mathcal{T} \not \subseteq  
	A^{-1}(W_{1}),$
	contradicting the fact 
	that $V(I)$ is the Zariski closure of $\mathcal{T}$. 
\end{proof}

\else
The proof of Proposition~\ref{prop: iv-prop}
can be found in~\cite{DBLP:journals/corr/abs-1201-5086}.
\fi
The following Proposition, which follows directly
from the definition of an inductive invariant, 
can serve as a sufficient
condition for a set of polynomials 
to be inductive invariants of a given loop.

\begin{Proposition}
\label{prop: inv-semi}
Let $\mathcal{L}$ be a loop  with variables $X$ and $m$
branches $(C_i, A_i)$ $i=1,\ldots,m$.
Let $P \subset \mathbb{Q}[X]$.
If $V(P)$ contains the initial values of $\mathcal{L}$,
and if for each  $\alpha \in V(P) \cap Z(C_i)$, we have $A_i(\alpha) \in V(P)$, 
then all the polynomials in $P$ are inductive polynomial equation 
invariants of $\mathcal{L}$.
\end{Proposition}
Note Proposition~\ref{prop: inv-semi} states a sufficient condition
for a set of polynomials to be invariant, not to generate the invariant
ideal. 
% Thus, one can not get a necessary and sufficient condition for a 
% set of polynomials to generate merely by combining Proposition~\ref{prop: inv-semi} 

We shall use Proposition~\ref{prop: inv-semi} as a criterion 
to certify  given polynomials
are indeed inductive invariants.
% \begin{proof}
% We only need to show that $\mathfrak{\mathcal{L}} \subseteq V(P)$.
% Denote $\mathfrak{T}_i$ by the values of the variable can reach at the 
% $i$'s execution. It is sufficient to show for all $i=0,1,\ldots, xxx$,
% we have  $\mathfrak{T}_i \subseteq V(P)$.
% This can be seen clearly by induction on the loop execution steps $k$.
% Case	$k=0$: clearly we have $\mathfrak{T}_0=V(I_0) \subseteq V(P)$ holds;
% Induction phase: assume for $i=n-1$, we have $\mathfrak{T}_i \subseteq V(P)$;
% 	    according to the fact that for each point $\alpha \in V(P)$, we have 
% $M(\alpha) \in V(P)$, we have $\mathfrak{T}_n \subseteq V(P)$. This completes
% the proof.
% \end{proof}
Actually, most loop invariant checking criteria
work in a similar spirit. The proposed criterion 
is more general than the various ``consecutions''
conditions in ~\cite{SSM04}, in the sense
that all invariants certifiable by those
``consecutions'' conditions is certifiable by the proposed criterion,
but there are invariants certifiable
by Proposition~\ref{prop: inv-semi}, which can not
be certified by any of the ``consecutions'' conditions.

\subsection{Implementation of the method}

We use polynomial interpolation to construct candidate 
invariants from a given template (which is 
either all possible dense polynomials 
up to a certain degree or a specific
form guessed by an oracle).
To do so, we need to take sufficiently
many points from the trajectory of the program execution.
This is done by emulating the program and recording the 
relevant values. 
To apply the criterion of Proposition~\ref{prop: inv-semi},
we need to compute the image of a
variety under a polynomial map. This is 
where we use state-of-art computer algebra 
software tools.

%%% \subsection{Implementation}
In this section, we describe two algorithms  for generating 
polynomial loop invariants that we have implemented.
We refer the first one as our {\em direct} method.

\begin{Notation}
    \label{nota: input}
Notations in the input of our algorithms:
    \begin{enumerate}[$(i)$]
% \begin{Notation}
\item $M:=m_1,m_2,\ldots,m_c$ is a sequence of monomials in the loop variables $X$
% \indent $X$ is the loop variables\\
\item $S := s_1, s_2, \ldots, s_r$ is a set 
of $r$ points on the inductive trajectory of the loop
\item $E$ is a polynomial system defining the loop initial values
\item $B$ is the transitions $(C_1, A_1), \ldots,  (C_m, A_m)$ of the loop
    \end{enumerate}
\end{Notation}

The subroutines in Algorithm~\ref{alg: pii} are explained as follows:
{\tt BuildLinSys}($M$, $S$) returns an $r \times c$ matrix $L$, 
such that $L_{i,j}$ is the evaluation 
of the $i$-th monomial in $M$ at the $j$-th 
point in $S$.
{\tt LinSolve}($L$) returns a matrix $N$ in row echelon form with full row rank, whose rows 
generate the null space of $L$ in $\mathbb{Q}^c$.
% regarded as linear operator $\mathbb{Q}^c$ 
%
%solution space of the linear system $L_{r \times c} \, (v_1, v_2, \ldots, v_c) = (0, 0, \ld.
%%
{\tt GenPoly}($M$, $\mathbf{v}$) returns the polynomial $\sum_{i=1}^c \; v_i\,m_i$, where
$\mathbf{v}=(v_1,v_2,\ldots,v_c)$ is a vector in $\mathbb{Q}^c$.

Note that  we can find effective tools for all the operations in Algorithm~\ref{alg: pii},
for instance,  we can find tools in~\cite{CLMLPX09} for computing the intersection of 
two constructible sets, or the image of a constructible set under a polynomial map as well as testing the inclusion relation.

However, there is a notable challenge with 
our direct algorithm (Algorithm~\ref{alg: pii}):
the coordinates of the points sampled on the trajectory
often grow dramatically in size.
This has clearly a negative impact on the solving of the linear system $L$.
All this leads to a severe memory
consumption issue, so we decided to consider 
an algorithm based on modular techniques.
We opted for a ``small prime'' approach, see Algorithm~\ref{alg: mii},
as we observed that many invariants of practical program loops have often
small coefficients.

Some additional subroutines, used in Algorithm~\ref{alg: mii}, 
are specified hereafter:
{\tt MaxMachinePrime}() returns the maximum machine-word prime;
{\tt PrevPrime}($p$) returns the largest prime less than $p$;
{\tt BuildLinSysModp}($M$, $S$, $p$) returns an $r \times c$ matrix $L$, 
such that $L_{i,j}$ is the evaluation of the $i$-th monomial in $M$ at the $j$-th 
point in $S$ modulo $p$.
{\tt LinSolveModp}($L$, $p$) returns a matrix $N$ 
in row echelon form with full row rank, whose rows 
generate the null space of $L$ in $\mathbb{Z}_p^c$.
{\tt RatRecon}($\mathbf{N}$, $\mathbf{P}$) returns a 
matrix $N$ with rational coefficients, such that
for each $i = 1 \ldots k$, the $i$-th matrix in $\mathbf{N}$
equal to the image of $N$ modulo the $i$-th prime in $\mathbf{P}$ if possible; otherwise returns {\rm FAIL}.

% regarded as linear operator $\mathbb{Q}^c$ 
%
%solution space of the linear system $L_{r \times c} \, (v_1, v_2, \ldots, v_c) = (0, 0, \ld.
%%
% {\tt Append}($G$, $e$) returns the new list by appending an object $e$
% to the end of $G$, where
% $G$ is a list of objects of the same type as $e$.

\begin{Proposition}
\label{prop: interp-inv-alg}
Both Algorithms~\ref{alg: pii} and~\ref{alg: mii} 
terminate for all inputs.
Moreover, when the output is not {\rm FAIL}, 
it is a list of polynomial equation invariants
for the target loop.
\end{Proposition}
\ifCoRR
\begin{proof}
The termination is easy to check 
since all loops iterate on finitely many terms and
each operation (sub-algorithm) does terminate. 
When the output is not {\rm FAIL}, 
that means the output satisfies the sufficient conditions
for polynomial invariants stated in 
Proposition~\ref{prop: inv-semi} 
and thus the conclusion follows from  Proposition~\ref{prop: inv-semi}.
\end{proof}
\else
The proof of Proposition~\ref{prop: interp-inv-alg}
appears in~\cite{DBLP:journals/corr/abs-1201-5086}.
\fi

\begin{Remark}
	We handle ``unlucky primes'' by checking the dimension of 
	the solution space (lines $15-18$ in Algorithm~\ref{alg: mii}): 
if the dimension of an image increases, then we drop this image;
		if a new image has a lower dimension, 
then we drop all previous images.
\ifCoRR
        Several points of Algorithm~\ref{alg: mii}) returns the same 
        ``{\rm FAIL}'' message, for sake a simplicity.
	However, we could customize the {\rm FAIL} message 
in each return point, for examples:
\begin{itemize}
\item the {\rm FAIL} at line $9$ implies that 
either the invariant ideal is the zero ideal 
	or the total degree of interpolated polynomials is too low; 
or the modulus is 
	too small; 
\item the {\rm FAIL} at line $23$ means that 
the product of the chosen moduli is still too 
	small and more images are needed.
\end{itemize}
	%
	% For sampling (or modular images of) points from the trajectory, use scattering points
	% as much as possible.
\fi
\end{Remark}

% For if the input variable involved in the assignment or
% initialization of the loop variables.

% For program with loops, we will build accompanion
% program to generate as many points as we needed in
% the variety of the inductive invariants.

% \input{algo1}

% !Tex root=mx-2012-casc.tex
\begin{algorithm}
\label{alg: pii}
% \dontprintsemicolon
% \linesnumbered
\caption{\pii($M$, $S$, $B$, $E$)}
\KwIn{
See Notation~\ref{nota: input} for $M$,  $S$, $B$, $E$
% \parbox{0.40\textwidth}{
% 	\begin{enumerate}
% 		\item[$X$ --] the list of loop variables
% 		\item[$B$ --] the list of transitions: $(C_i, A_i)$ $i=1\ldotsm$ %the list of loop variables;
% 		\item[$M$ --] the monomial support of polynomials: $m_1,m_2, \cdots,m_c$ %the list of loop variables;
% 		% \item[sampl --] a method to sample points on the trajectory of the loop; %the list of loop variables;
% 		% \item[td --] the maximal total degree of polynomials to interpolate; %the list of loop variables;
% 		% \item[npts --] the number of points to interpolate; %the list of loop variables;
% 		\item[$S$] the initial values of the loop. %the list of loop variables;
% 		\item[$E$ --] the polynomial system defining the initial values %the list of loop variables;
% 	 \end{enumerate}
% 	 }
 }
\KwOut{A set of polynomial inductive invariants of the target loop}  
% with support monomials in $M$ or {\rm FAIL}}
% Sample $npts$ points, $S$, with $sampl$ from the trajectory\;
% Write down the dense polynomial $f$ of total degree td w.r.t. vars\; 
% $L := \emptyset$\;
% \ForEach{point $pt \in S$}{
%	Specialize $vars$ in $f$ by $pt$ and get a linear equation 
%	$l$ w.r.t. the indeterminate coefficients of $f$\;	
%	$L := L \, \cup \, \{ l \} $\;
% }

$L := {\tt BuildLinSys}(M, S)$\; 

$N := {\tt LinSolve(L)}$ \;

$F := \emptyset$\;
\ForEach{row vector $\mathbf{v} \in N$}{
$F := F \, \cup \, \{ {\tt GenPoly}(T, \mathbf{v}) \} $\;
}

\lIf{$Z(E) \not \subseteq V(F)$}{
	\Return {\rm FAIL}\;
}

\ForEach{$(C_i, A_i) \in B$}{
       % $W := {\tt PolyMapImg}(A, V(F) \cap Z(C))$\;
	   % Compute the (Zariski closure of) the image $W$ of $V(c) \cap V(F)$ mapped by $A$\;
	   \lIf{$A_i(V(F) \cap Z(C_i)) \; \not \subseteq \; V(F)$}{
	    \Return {\rm FAIL}\;
    }
}
\Return $F$\;
\end{algorithm}

% \subsection{Complexity issues}
Note that Algorithm~\ref{alg: pii} 
and Algorithm~\ref{alg: mii} 
will sometimes return {\rm FAIL} even if the bounds  for the 
polynomial degrees and coefficient sizes are known.
When these algorithms return a list of non-trivial polynomials,
we are not sure whether those polynomial can generate the whole
loop invariant ideal or not. 
However, in practice, these algorithms often find meaningful results
quickly.  Indeed, both algorithms run 
in singly exponential time w.r.t. number of variables 
for a fix total degree bound, which is stated formally as below.

\begin{Proposition}
Algorithm~\ref{alg: mii} runs in singly exponential time
w.r.t. number of loop variables.
\end{Proposition}
\begin{proof}
The complexity of Algorithm~\ref{alg: mii} between Lines 1 and 26
is polynomial in the number of monomials in the support.

The number of those monomials is singly exponential t
w.r.t. number of loop variables.
In addition, applying our criterion to 
certify the result (Lines 27 to 30)
can be reduced  to an ideal membership problem, which is singly exponential 
w.r.t. number of loop variables.% see for example in~\cite{MayrMeyer}.
\end{proof}

\ifCoRR
% \fbox{The following is not clear enough, but appears only in the CoRR version}
In particular, if the total degree bound 
supplied is greater of equal than the degree of 
invariant ideal and the sample points
are sufficiently many, 
% ,  we compute the interpolation polynomials 
% to total degree 
% (on sufficiently many points)
%  beyond the degree bound in Theorem~\ref{thm: main}, 
then with a high possibility 
(depending on the selection of sample points and also 
on the choice of the moduli for Algorithm~\ref{alg: mii}), 
a list of polynomials generating the invariant ideal will be computed by
out method.
\fi

\subsection{Experimental results}

We have applied Algorithm~\ref{alg: mii} to 
the example programs used in the paper~\cite{RK07a}, 
and we are able to find the loop invariants by trying 
total degree up to $4$ for most loops within $60$ seconds. 
See the Table~\ref{tab: loopinv} for details.

% \begin{table}
%     \caption{Experiments on selected programs}
%     \vspace{1ex}
%     \begin{center}
%     \begin{tabular}{||l | c | c | c | c | c | c | c | c | c||}
%         \hline
%         prog\footnote{For more details, see \url{www.csd.uwo.ca/\~rong/loop\_inv.tgz} for the source of all the programs.}        & 1     & 2    & 3    & 4    & 5    & 6   & 7 & 8 & 9\\
%         \# vars     & 6     & 4    & 5    & 5    & 6    & 3   & 4	& 3	& 4\\
%         degree      & 3     & 3    & 4    & 3    & 3    & 3   & 6	& 3	& 2\\
%         tm (s)      & 10.9  & 0.6  & 3.74 & 1.4  & 3.1  & 0.2 & 12	& 0.4 &0.06\\
%         \hline
%     \end{tabular}
%     \end{center}
%     \label{tab: loopinv}
% \end{table}
% knuth & cohencu   & fermat_knuth & prodbin   & kr07  & kov07   & wwetal & wensley & cubic_root

% \input{algo2}

% !Tex root=mx-2012-casc.tex
\begin{algorithm}
\label{alg: mii}
% \dontprintsemicolon
% \linesnumbered
\caption{\mii($M$, $S$, $B$, $E$, $n$)}
\KwIn{See Notation~\ref{nota: input} for $M$, $S$,  $B$, $E$; $n$ is the maximal number of modular images to use}
\KwOut{A set of polynomial inductive invariants of the target loop}
$p:= {\tt MaxMachinePrime}()$\; % to the largest prime fit the machine integer size\;
$L := {\tt BuildLinSysModp}(M, S, p)$\; 
$N := {\tt LinSolveModp(L, p)}$ \;
$d := \dim(N)$ \;
$\mathbf{N} :=(N)$\;
$\mathbf{P} :=(p)$\;
$i := 1$\;

\While{$i \leq n$ and  $p>2$}{

	\lIf{$d=0$}{
	\Return {\rm FAIL}\;
	}

    $N := {\tt RatRecon}(\mathbf{N},  \mathbf{P})$\;
	\lIf{$N \neq {\rm FAIL}$}{
		{\bf break}\;	
	}	
    
    $p := {\tt PrevPrime}(p)$\;
	
    $L := {\tt BuildLinSysModp}(M, S, p)$\; 

    $N := {\tt LinSolveModp(L, p)}$ \;
	\If{$d > \dim(N)$}
	{
		$d := \dim(N)$;
		$\mathbf{N} :=(N)$\;
        $\mathbf{P} :=(p)$\;
        $i := 1$\;
	}
	\ElseIf{$ d = \dim(B)$}{
         $\mathbf{N}:={\tt Append}(\mathbf{N},N)$\;
		$\mathbf{P}  := {\tt Append}(\mathbf{P},p)$\;
		$i := i+1$\;
	}
}

\lIf{$i>n$ or $p = 2$}{
	\Return {\rm FAIL}\;
}

$F := \emptyset$\;
\ForEach{row vector $\mathbf{v} \in N$}{
$F := F \, \cup \, \{ {\tt GenPoly}(T, \mathbf{v}) \} $\;
}

\lIf{$Z(E) \not \subseteq V(F)$}{
	\Return {\rm FAIL}\;
}

\ForEach{$(C_i, A_i) \in B$}{
       % $W := {\tt PolyMapImg}(A, V(F) \cap Z(C))$\;
	   % Compute the (Zariski closure of) the image $W$ of $V(c) \cap V(F)$ mapped by $A$\;
	   \lIf{$A_i(V(F) \cap Z(C_i)) \; \not \subseteq \; V(F)$}{
	    \Return {\rm FAIL}\;
    }
}
\Return $F$\;
\end{algorithm}

% coupled5,5-1. our 0.671; 9.45 we are better at those with irrational eigenvalues, when # of variables are smaller
% non-inv2. our: 1.20;  ai: 3.83
%> st :=time(); TRDmodp_loop_invariant_da\                                      
%> ([x,y,z,tz],[[z,x+z,y+z,z+y]],{[1,2,3,3]},total_degree=4);
%                                  st := 2.300
%
%memory used=168.0MB, alloc=80.2MB, time=2.88
%          3    2        2         2                  2    3    2          2
%[z - tz, x  + x  y + 3 x  tz - x y  - 4 x y tz - x tz  + y  + y  tz - y tz
%
%         3
%     + tz  - 4]
%
%> time()-st;                                                                   
%                                     0.760

%> st :=time(); our := TRDmodp_loop_invar\                                      
%> iant_da(vars,transitions,{[1,2,1,1]},total_degree=4);
%                                  st := 4.460
%
%memory used=320.6MB, alloc=80.2MB, time=5.38
%                     2            2       3          4              2
%our := [x - z, -1 + z  - 2 y z + y , 4 y z  + 1 - 3 z  + 4 y z + 2 z  - 16 a,
%
%     5         2       3
%    z  - 20 y z  + 10 z  + 64 y a - 80 z a - 4 y - 11 z]
%
%> time()-st;
%                                     1.180

In the following table, we supply experimental results
for computing absolute inductive invariants for some 
well-known programs from literature as well as some 
homemade examples  marked with a star $*$.
The first column labeled by ``\# vars'' is the number of loop variables;
the second column  labeled by ``deg'' is the total degree tried for the methods 
which use a degree bound; the third column labeled by ``PI'' is the timing 
of the our method; the fourth column labeled by ``AI'' is the timing of the
method described in~\cite{RK07b}; 
the fifth column labeled by ``IF'' is the timing of the
method described in~\cite{RK07a}; 
the sixth column labeled by ``ALIGATOR'' is the timing of the
method described in~\cite{Kov08}. 
The time unit is second;
the ``NA'' symbol in a time field means that the related method does
support the input program;  
% the ``TO'' symbol in a time field means 
% that the related computation is not complete in the time limit (10 minutes);
the ``FAIL'' symbol in a time field means that the output is not ``correct''. 
All the tests were done 
% with {\sc Maple 16}.
% , while the computations on multiplicative relation lattice were done on the same machine with GAP 4.4.12 + Alnuth 2.3.1 + KASH 2.5,  
using an Intel Core 2 Quad {\small CPU} 2.40{\small GHz} 
with 8.0{\small GB} memory.

% ALIGATOR fails to justify P-solvability and given wrong results on many an example
% See ALIGATOR-demo 

% PI: polynomial interpolation
% AI: abstract interpretation
% IF: ideal fix point
% SE: solving and elimination  
\begin{table}
    \caption{Experiments on selected programs}
    \vspace{1ex}
    \begin{center}
    \begin{tabular}{||l |  c | c | c | c|c |c||}
        \hline
		prog.\footnotemark[1]	&  \# vars	& deg & PI & AI & FP &SE \\
		\hline
		cohencu	&	4		& 	3	 & 0.6  & 0.93 &  0.28 & 0.13 \\
		cohencu	&	4		& 	2	 & 0.06 & 0.76 & 0.28 & 0.13  \\    
		fermat	&	5		& 	4	 & 3.74 & 0.79 & 0.37 & 0.1  \\  
		prodbin	&	5		& 	3	 & 1.4  & 0.74 & 0.36 & 0.13 \\  
		rk07	&	6		& 	3	 & 3.1  & 2.23 & NA & 0.35  \\  
		kov08	&	3		& 	3	 & 0.2  & 0.57 & 0.22 & 0.01 \\  
		sum5	&	4		& 	5	 & 12 & 1.60 & 2.25 & 0.16\footnotemark[2]  \\  
		wensley2	&	3		& 	3	 & 0.4 & 0.84 & 0.39 & 0.21  \\  
		int-factor 	&	6		& 	3	 & 60.9 & 1.28 & 160.7 & 0.9\\
		fib(coupled) 	&	4		& 	4	 & 2.4 & 0.71 & NA & NA\\
		fib(decoupled) 	&	6		& 	4	 & 4.3 & 1.28 & 160.7 & FAIL\\
		non-inv2* &	4		& 	3	 & 1.2 & 3.83 & NA & FAIL\\
		coupled-5-1* &	4		& 	4	  & 1.1 & 9.58 & NA & NA\\
		coupled-5-2* &	5		& 	4	  & 5.38 & 15.8 & NA & NA\\
		% mannadiv 	&	3		& 	3	 & 0.23 & 1.16 & NA & NA\\
		mannadiv 	&	3		& 	3	 & 0.1 & 0.83 & NA & 0.04\\
        \hline
    \end{tabular}
    \end{center}
    \label{tab: loopinv}
\end{table}

\footnotetext[1]{For more details, see \url{http://www.csd.uwo.ca/~rong/loop_inv.tgz} for the source of all the programs.}
% \footnotetext[3]{The initial condition of the target loop in this program is not semi-algebraic}

% knuth & cohencu   & fermat_knuth & prodbin   & kr07  & kov07   & wwetal & wensley & cubic_root

\footnotetext[2]{There might be a bug in the version of Aligator we are using, because the computation can not finished in 1hr
		in this test; the timing was reported by Laura Kovacs in a demo of Aligator. }

\ifCoRR
The following example shows how we can use the degree and dimension information:w to assure that we are computing
the whole invariant ideal.
\begin{Example}
\label{ex: tricky}
Consider the following recurrence relation
on $(x,y,z)$:
\begin{equation*}
\left(
\begin{array}{c}
x(n+1) \\
y(n+1) \\
z(n+1) 
\end{array}
\right)
\;
=
\;
\left(
\begin{array}{c c r}
0 & 0 & 1 \\
1 & 0 & -3 \\
0 & 1 & 3
\end{array}
\right)
{}
\,
{}
\left(
\begin{array}{c}
x(n) \\
y(n) \\
z(n) 
\end{array}
\right)
{}
\end{equation*}
with initial value $\left( x(0),y(0),z(0) \right)= (1,2,3)$. 
Denote by $M$ the coefficient matrix.
Note that the characteristic polynomial of $M$ has $1$ as a triple root
and the multiplicative relation ideal of the eigenvalues is zero-dimensional. 
So the invariant ideal of this recurrence has dimension either $0$ or $1$.
On the other hand, we can show that for all $k \in \mathbb{N}$, 
we have $M^k \neq M$;
so there are infinitely many points in the set 
$\{ (x(k), y(k), z(k)) \mid k \in \mathbb{N}\}$,\
whenever $\left( x(0),y(0),z(0) \right) \neq (0,0,0)$.

With our method, we are able to  compute the following invariant polynomials
  $$x + y + z - 6, y^2  + 4 y z + 4 z^2  - 6 y - 24 z + 20,$$
  which generate a prime ideal of dimension $1$ 
(thus the invariant ideal of this recurrence), in less than $0.25$s.
% M := Matrix([[0,0,1],[1,0,-3],[0,1,3]]);
% f := CharacteristicPolynomial(M,x);
% f := (x-1)^3;
% st := time(); TRDrecurrence_invariants_interpolate([x,y,z],[z,x-3*z, y+3*z], [1,2,3]); time()-st;
% 0.249
% 
% 
% > vars :=[x,y,z,a,b,c];
%                           vars := [x, y, z, a, b, c]
% 
% > inits := [a,b,c];
%                               inits := [a, b, c]
% 
% > initials := TRDgrid_initial_values([a,b,c], [a,b,c], 4);           
% 
% > TRDmodp_loop_invariant_da(vars, [[z,x-\                                      
% 3*z, y+3*z,a,b,c]], initials,total_degree=4); # 38s
% 
% > TRDmodp_loop_invariant_da([x,y,z], [[z\                                      
%  ,x-3*z, y+3*z]], {[1,0,1]},total_degree=4);
%
\end{Example}
\fi

% The following example is taken from~\cite{KZ08}, where 
% the closed form of the sequences are used to compute the algebraic
% relations among them.
% \begin{Example}
% The Tribonacci sequence is defined by recurrence equation 
% $$T(n+3) = T(n) + T(n+1) + T(n+2),$$
% with initial values $T(0)=0, T(1)=1, T(2) =1$. 
% 
% We can apply Algorithm~
% Within $0.5$s,  algorithm  
% \end{Example}
% > st := time(); rel := TRDtribonacci_alg\                                      
% > rel([x,y,z,u,v,w],[1,2,3,4,5,6],total_degree=4, num_points=200); time() -st;
%         0.48 for degree 3, 1.
%  st := time(); rel := TRDtribonacci_alg\                                      
% > rel([x,y,z,u,v,w],[1,2,3,4,5,6],total_\                                      
% > degree=6, num_points=400,num_of_primes=4): time() -st;

% \input{conclusion}

% !Tex root=mx-2012-casc.tex
\section{Concluding remarks}
%
% \fbox{Mention the difference than related work}
%
% \fbox{Completeness}
In this article, we propose
a loop invariant computing method based on polynomial
interpolation. We supply a sharp
total degree bound for polynomials generating the loop invariant of 
 $P$-solvable recurrences.
We supply also sufficient conditions for 
inductive loop invariant to be trivial or non trivial.
% though in theory (mainly dimension related and 
% degree bound), we are able to give complete
% answer to invariant ideal problem in most 
% cases we have computed right, using this 
% polynomial interpolation based method.
%In a future work, we would like to estimate a bound on the coefficient size,
%for the polynomials with total degree less than that sharp degree bound
%generating the invariant ideal of a  $P$-solvable recurrence.
% in a chosen representation (e.g. by regular chains or a Gr\"ober basis),
% so that we can make the theory more complete. 

The current implementation is for dense interpolation.
% in this case,
% the interpolation template is a dense polynomial with all 
% coefficients indeterminate.
% However, the method can can used to find 
% polynomial equation invariant fit any template without modification. 
However, we observe that for loops with sparse polynomials in
the assignments, the computed invariants are often sparse too.
As future work, we will investigate suitable sparse 
interpolation techniques
for interpolating polynomial loop invariant.  

{\bf Acknowledgement.} We thank Laura Kov\'avs and Enric Rodr\'iguez-Carbonell for providing the implementations of their
invariant generation methods for our test.
% \nocite{Aligator} 


\begin{thebibliography}{10}

\bibitem{CJ07}
Jacques Carette and Ryszard Janicki.
\newblock Computing properties of numerical imperative programs by symbolic
  computation.
\newblock {\em Fundam. Inf.}, 80:125--146, January 2007.

\bibitem{CLMLPX09}
Changbo Chen, Fran\c{c}ois Lemaire, Liyun Li, Marc {Moreno Maza}, Wei Pan, and
  Yuzhen Xie.
\newblock Computing with constructible sets in maple, 2009.
\newblock Submitted to J. of Symbolic Computation.

\bibitem{CXYZ07}
Yinghua Chen, Bican Xia, Lu~Yang, and Naijun Zhan.
\newblock Generating polynomial invariants with discoverer and qepcad.
\newblock In {\em Formal Methods and Hybrid Real-Time Systems}, pages 67--82,
  2007.

\bibitem{CLO97}
D.~Cox, J.~Little, and D.~O'Shea.
\newblock {\em Using Algebraic Geometry}.
\newblock Graduate Text in Mathematics, 185. Springer-Verlag, New-York, 1998.

\bibitem{RK07b}
D.~Kapur E.~Rodriguez-Carbonell.
\newblock Automatic generation of polynomial invariants of bounded degree using
  abstract interpretation.
\newblock {\em Science of Computer Programming}, 64(1):54--75, 2007.

\bibitem{GG99}
{J. von zur} Gathen and J.~Gerhard.
\newblock {\em Modern Computer Algebra}.
\newblock Cambridge University Press, 1999.

\bibitem{Ge93}
G.~Ge.
\newblock {\em Algorithms related to multiplicative representations of
  algebraic numbers}.
\newblock PhD thesis, U.C. Berkeley, 1993.

\bibitem{Heintz83}
Joos Heintz.
\newblock Definability and fast quantifier elimination in algebraically closed
  fields.
\newblock {\em Theor. Comput. Sci.}, pages 239--277, 1983.

\bibitem{Kap05}
Deepak Kapur.
\newblock Automatically generating loop invariants using quantifier
  elimination.
\newblock In {\em Deduction and Applications}, 2005.

\bibitem{Karr76}
Michael Karr.
\newblock Affine relationships among variables of a program.
\newblock {\em Acta Inf.}, 6:133--151, 1976.

\bibitem{KZ08}
Manuel Kauers and Burkhard Zimmermann.
\newblock Computing the algebraic relations of c-finite sequences and
  multisequences.
\newblock {\em J. Symb. Comput.}, 43:787--803, November 2008.

\bibitem{Kov08}
Laura Kov\'{a}cs.
\newblock Invariant generation for p-solvable loops with assignments.
\newblock In {\em Proceedings of the 3rd international conference on Computer
  science: theory and applications}, CSR'08, pages 349--359, Berlin,
  Heidelberg, 2008. Springer-Verlag.

\bibitem{KV09b}
Laura Kov\'acs and Andrei Voronkov.
\newblock Interpolation and symbol elimination.
\newblock In Renate Schmidt, editor, {\em Automated Deduction – CADE-22},
  volume 5663 of {\em Lecture Notes in Computer Science}, pages 199--213.
  Springer Berlin / Heidelberg, 2009.

\bibitem{MS04a}
Markus M\"{u}ller-Olm and Helmut Seidl.
\newblock Computing polynomial program invariants.
\newblock {\em Inf. Process. Lett.}, 91(5):233--244, September 2004.

\bibitem{MS04b}
Markus M\"uller-Olm and Helmut Seidl.
\newblock A {N}ote on {K}arr's {A}lgorithm.
\newblock In Josep D\'iaz, Juhani Karhum\"aki, Arto Lepist\"o, and Donald
  Sannella, editors, {\em Automata, Languages and Programming}, volume 3142 of
  {\em Lecture Notes in Computer Science}, pages 1016--1028, Turku, Finland,
  July 2004. Springer.

\bibitem{Osb00}
Martin~J. Osborne.
\newblock Math tutorial: first-order difference equations, 2000.

\bibitem{RK04b}
E.~Rodr{\'\i}guez-Carbonell and D.~Kapur.
\newblock {An Abstract Interpretation Approach for Automatic Generation of
  Polynomial Invariants}.
\newblock In {\em {International Symposium on Static Analysis (SAS 2004)}},
  volume 3148 of {\em Lecture Notes in Computer Science}, pages 280--295.
  Springer-Verlag, 2004.

\bibitem{RK04}
E.~Rodr\'{\i}guez-Carbonell and D.~Kapur.
\newblock Automatic generation of polynomial loop invariants: Algebraic
  foundations.
\newblock ISSAC '04, pages 266--273. ACM, 2004.

\bibitem{RK04a}
E.~Rodr{\'\i}guez-Carbonell and D.~Kapur.
\newblock {Program Verification Using Automatic Generation of Invariants}.
\newblock In {\em {1st International Colloquium on Theoretical Aspects of
  Computing (ICTAC'04)}}, volume 3407 of {\em LNCS}, pages 325--340.
  Springer-Verlag, 2005.

\bibitem{RK07a}
Enric Rodr\'{\i}guez-Carbonell and Deepak Kapur.
\newblock Generating all polynomial invariants in simple loops.
\newblock {\em J. Symb. Comput.}, 42(4):443--476, 2007.

\bibitem{SSM04}
Sriram Sankaranarayanan, Henny~B. Sipma, and Zohar Manna.
\newblock Non-linear loop invariant generation using gr\"obner bases.
\newblock {\em SIGPLAN Not.}, 39:318--329, 2004 2004.

\bibitem{TRSS01}
Ashish Tiwari, Harald Rue\ss, Hassen Sa\"{\i}di, and Natarajan Shankar.
\newblock A technique for invariant generation.
\newblock In {\em Proceedings of the 7th International Conference on Tools and
  Algorithms for the Construction and Analysis of Systems}, TACAS 2001, pages
  113--127, London, UK, UK, 2001. Springer-Verlag.

\end{thebibliography}
\end{document}